\documentclass[11pt,letterpaper]{article}
\usepackage{fullpage}
\usepackage{amsfonts,amssymb,amsmath,amsthm}
\usepackage{enumerate}
\usepackage{graphicx}
\usepackage{color}
\usepackage[dvips,
pdfstartview=FitH,
pdfpagemode=None,
colorlinks=true,
citecolor=blue,
linkcolor=blue]{hyperref}

\theoremstyle{plain}
\newtheorem{thm}{Theorem}[section]
\newtheorem{lem}{Lemma}[section]
\newtheorem{cor}{Corollary}[section]

\theoremstyle{definition}
\newtheorem{defn}{Definition}[section]
\theoremstyle{remark}
\newtheorem{rem}{Remark}

\newcommand{\thmref}[1]{Theorem~\ref{#1}}

\newcommand{\lemref}[1]{Lemma~\ref{#1}}
\newcommand{\corref}[1]{Corollary~\ref{#1}}
\newcommand{\defnref}[1]{Definition~\ref{#1}}

\newcommand{\figref}[1]{Figure~\ref{#1}}

\newcommand{\mylemma}[2]{\begin{lem}\label{#1}#2\end{lem}}
\newcommand{\mylem}[3]{\begin{lem}[#3]\label{#1}#2\end{lem}}
\newcommand{\againlemma}[2]{\noindent\textbf{Lemma~\ref{#1}}
    (from page \pageref{#1})\textbf{.}\emph{#2}}

\newcommand{\mytheorem}[2]{\begin{thm}\label{#1}#2\end{thm}}
\newcommand{\mythm}[3]{\begin{thm}[#3]\label{#1}#2\end{thm}}

\newcommand{\againtheorem}[2]{\noindent\textbf{Theorem~\ref{#1}}
    (from page \pageref{#1})\textbf{.}\emph{#2}}

\newcommand{\defeq}{\stackrel{\text{\tiny def}}{=}}
\newcommand{\supp}{\mathrm{supp}}
\newcommand{\E}{\mathbb{E}}
\newcommand{\lb}{\left(}
\newcommand{\rb}{\right)}

\usepackage{tikz}
\usetikzlibrary{
  arrows,
  positioning,
  scopes,
  calc,
  fit,
}
\tikzset{
font=\footnotesize,
node distance=5mm,
every fit/.append style=text badly centered,
component/.style = {
	rounded corners=3mm, minimum width=1.5cm, minimum height=1cm, thick, draw=black,
	},
point/.style={coordinate},
txt/.style={minimum width=1cm, minimum height=1cm},
post/.style={rounded corners, >=stealth, ->, thick, black},
pre/.style={rounded corners, >=stealth, <-, thick, black},
ord/.style={rounded corners, thick, black}
}

\title{Randomness-Efficient Curve Samplers}
\author{Zeyu Guo\footnote{\texttt{zguo@caltech.edu}. Supported by
NSF CCF-1116111, NSF CCF-1038578 and BSF 2010120.}\\
Computer Science Department\\
California Institute of Technology\\ Pasadena, CA 91125}

\begin{document}

%\begin{titlepage}
%\clearpage
\maketitle

\begin{abstract}
Curve samplers are sampling algorithms that proceed by viewing the domain as a vector space over a finite field, and randomly picking a low-degree curve in it as the sample. Curve samplers exhibit a nice property besides the sampling property: the restriction of low-degree polynomials over the domain to the sampled curve is still low-degree. This property is often used in combination with the sampling property and has found many applications, including PCP constructions, local decoding of codes, and algebraic PRG constructions.

The randomness complexity of curve samplers is a crucial parameter for its applications. It is known that (non-explicit) curve samplers using $O(\log N+\log(1/\delta))$ random bits exist, where $N$ is the domain size and $\delta$ is the confidence error. The question of explicitly constructing randomness-efficient curve samplers was first raised in \cite{TU06} where they obtained curve samplers with near-optimal randomness complexity.

We present an explicit construction of low-degree curve samplers with {\em optimal} randomness complexity (up to a constant factor), sampling curves of degree $\left( m\log_q(1/\delta)\right)^{O(1)}$ in $\mathbb{F}_q^m$. Our construction is a delicate combination of several components, including extractor machinery, limited independence, iterated sampling, and list-recoverable codes.
\end{abstract}

%\end{titlepage}

\section{Introduction}

Randomness has numerous uses in computer science, and sampling is one of its most classical applications: Suppose we are interested in the size of a particular subset $A$ lying in a large domain $D$. Instead of counting the size of $A$ directly by enumeration, one can randomly draw a small sample from $D$ and calculate the density of $A$ in the sample. The approximated density is guaranteed to be close to the true density with probability $1-\delta$ where $\delta$ is very small, known as the {\em confidence error}. This sampling technique is extremely useful both in practice and in theory.

One class of sampling algorithms, known as {\em curve samplers}, proceed by viewing the domain as a vector space over a finite field, and picking a random low-degree curve in it.
%As the degree of curves and the size of the finite field varies, this method yields sampling algorithms with different parameters (randomness, error, sample size, etc).
Curve samplers exhibit the following nice property besides the sampling property: the restriction of low-degree polynomials over the domain to the sampled curve is still low-degree. This special property, combined with the sampling property, turns out to be useful in many settings, e.g local decoding of Reed-Muller codes and hardness amplification \cite{STV01}, PCP constructions \cite{AS98, ALMSS98, MR08}, algebraic constructions of pseudorandom-generators \cite{SU05,Uma03}, extractor constructions \cite{SU05, TU06}, and some pure complexity results (e.g. \cite{SU06}).

The problem of constructing explicit low-degree curve samplers was raised in \cite{TU06}. 
More specifically, we are looking for low-degree curve samplers with small sample complexity (polylogrithmic in the domain size) and confidence error (polynomially small in the domain size), and we focus on minimizing the randomness complexity.
%Randomness complexity is the central parameter of curve samplers.
The simplest way is picking a completely random low-degree curve whose sampling properties are guaranteed by tail bounds for limited independence. The randomness complexity of this method, however, is far from being optimal. The probabilistic method guarantees the existence of (non-explicit) low-degree curve samplers using $O(\log N+\log(1/\delta))$ random bits where $N$ is the domain size and $\delta$ is the confidence error. The real difficulty, however, is to find an explicit construction matching this bound. 
%In this paper, we present explicit constructions of randomness-efficient curve samplers.

\subsection{Previous work}

Randomness-efficient samplers (without the requirement that the sample points form a curve) are constructed in \cite{CG89,Gil98,BR94,Zuc97}. In particular, \cite{Zuc97} obtains explicit samplers with optimal randomness complexity (up to a $1+\gamma$ factor for arbitrary small $\gamma>0$) using the connection between samplers and extractors. See \cite{G11} for a survey of samplers.

Degree-1 curve samplers are also called {\em line samplers}.
Explicit randomness efficient line samplers are constructed in the PCP literature \cite{BSVW03, MR08}, motivated by the goal of constructing almost linear sized PCPs. In \cite{BSVW03} line samplers are derandomized by picking a random point and a direction sampled from an $\epsilon$-biased set, instead of two random points. An alternative way is suggested in \cite{MR08} where directions are picked from a subfield.
It is not clear, however, how to apply these techniques to higher degree curves.

In \cite{TU06} it was shown how to explicitly construct derandomized curve samplers with near-optimal parameters. Formally they obtained
\begin{itemize}
\item curve samplers picking curves of degree $\lb\log\log N+\log(1/\delta)\rb^{O(\log\log N)}$ using $O(\log N+\log(1/\delta)\log\log N)$ random bits, and
\item curve samplers picking curves of degree $\lb\log (1/\delta)\rb^{O(1)}$ using  $O(\log N+\log(1/\delta)(\log\log N)^{1+\gamma})$ random bits for any constant $\gamma>0$
\end{itemize}
for domain size $N$, field size $q\geq (\log N)^{\Theta(1)}$ and confidence error $\delta=N^{-\Theta(1)}$.
Their work left the problem of explicitly constructing low-degree curve samplers (ideally picking curves of degree $O(\log_q(1/\delta))$) with essentially optimal $O(\log N+\log(1/\delta))$ random bits as a prominent open problem.

\subsection{Main results}

We present an explicit construction of low-degree curve samplers with optimal randomness complexity (up to a constant factor).
In particular, we show how to sample degree-$\lb m\log_q(1/\delta)\rb^{O(1)}$ curves in $\mathbb{F}_q^m$ using $O(\log N+\log(1/\delta))$ random bits for domain size $N=|\mathbb{F}_q^m|$ and confidence error $\delta=N^{-\Theta(1)}$.
Before stating our main theorem, we first present the formal definition of samplers and curve samplers.

\paragraph{Samplers.}
Given a finite set $\mathcal{M}$ as the domain, the {\em density} of a subset $A\subseteq \mathcal{M}$ is $\mu(A)\defeq\frac{|A|}{|\mathcal{M}|}$. For a collection of elements $\mathcal{T}=\{t_i: i\in I\}\in\mathcal{M}^I$ indexed by set $I$, the density of $A$ in $\mathcal{T}$ is $\mu_\mathcal{T}(A)\defeq\frac{|A\cap\mathcal{T}|}{|\mathcal{T}|}=\Pr_{i\in I}[t_i \in A]$.

\begin{defn}[sampler]
A {\em sampler} is a function $S: \mathcal{N}\times \mathcal{D}\to \mathcal{M}$ where $|\mathcal{D}|$ is its {\em sample complexity} and $\mathcal{M}$ is its {\em domain}.
We say $S$ samples $A \subseteq \mathcal{M}$ with {\em accuracy error} $\epsilon$ and {\em confidence error} $\delta$ if $\Pr_{x\in \mathcal{N}} [ |\mu_{S(x)}(A) - \mu(A)| > \epsilon] \leq \delta$ where $S(x)\defeq\{S(x,y): y\in \mathcal{D}\}$. We say $S$ is an $(\epsilon,\delta)$ sampler if it samples all subsets $A \subseteq \mathcal{M}$ with accuracy error $\epsilon$ and confidence error $\delta$. The {\em randomness complexity} of $S$ is $\log(|\mathcal{N}|)$.
\end{defn}

\begin{defn}[curve/line sampler]
Let $\mathcal{M}=\mathbb{F}_q^D$ and $\mathcal{D}=\mathbb{F}_q$.
The sampler $S: \mathcal{N}\times \mathcal{D} \to \mathcal{M}$ is a {\em degree-$t$ curve sampler} if for all $x\in \mathcal{N}$, the function $S(x,\cdot): \mathcal{D}\to \mathcal{M}$ is a curve (see \defnref{defn_manifold}) of degree at most $t$ over $\mathbb{F}_q$.
When $t=1$, $S$ is also called a {\em line sampler}.
\end{defn}

\begin{thm}[main]\label{thm_sampler_1}
For any $\epsilon,\delta>0$, integer $m\geq 1$, and sufficiently large prime power $q\geq\lb\frac{m \log(1/\delta)}{\epsilon}\rb^{\Theta(1)}$, there exists an explicit degree-$t$ curve sampler for the domain $\mathbb{F}_q^m$ with $t=\lb m\log_q(1/\delta)\rb^{O(1)}$, accuracy error $\epsilon$, confidence error $\delta$, sample complexity $q$, and randomness complexity $O\lb m\log q+\log (1/\delta)\rb=O(\log N+\log(1/\delta))$ where $N=q^m$ is the domain size.
Moreover, the curve sampler itself has degree $\lb m\log_q(1/\delta)\rb^{O(1)}$ as a polynomial map.
\end{thm}

\thmref{thm_sampler_1} has better degree bound and randomness complexity compared with the constructions in \cite{TU06}. We remark that the degree bound, being $\lb m\log_q(1/\delta)\rb^{O(1)}$, is still sub-optimal compared with the lower bound $\log_q(1/\delta)$ (see Appendix \ref{deglb} for the proof of this lower bound). However in many cases it is satisfying to achieve such a degree bound. 

As an example, consider the following setting of parameters: domain size $N=q^m$, field size $q=\lb\log N\rb^{\Theta(1)}$, confidence error $\delta=N^{-\Theta(1)}$, and accuracy error $\epsilon=\lb\log N\rb^{-\Theta(1)}$.
Note that this is the typical setting in PCP and other literature \cite{ALMSS98,AS98,STV01,SU05}.
In this setting, we have the following corollary in which the randomness complexity is logarithmic and the degree is polylogarithmic.

\begin{cor}\label{cor_sampler}
Given domain size $N=|\mathbb{F}_q^m|$, accuracy error $\epsilon=\lb\log N\rb^{-\Theta(1)}$, confidence error $\delta=N^{-\Theta(1)}$, and large enough field size $q=\lb\log N\rb^{\Theta(1)}$, there exists an explicit degree-$t$ curve sampler for the domain $\mathbb{F}_q^m$ with
accuracy error $\epsilon$, confidence error $\delta$, randomness complexity $O(\log N)$, sample complexity $q$, and $t\leq\lb\log N\rb^c$ for some constant $c>0$ independent of the field size $q$.
\end{cor}

It remains an open problem to explicitly construct curve samplers that have optimal randomness complexity $O(\log N+\log(1/\delta))$ (up to a constant factor), and sample curves with optimal degree bound $O(\log_q(1/\delta))$. It is also an interesting problem to achieve the optimal randomness complexity up to a $1+\gamma$ factor for any constant $\gamma>0$ (rather than just an $O(1)$ factor), as achieved by \cite{Zuc97} for general samplers. The standard techniques as in \cite{Zuc97} are not directly applicable as they increase the dimension of samples and only yield $O(1)$-dimensional manifold samplers.

\subsection{Techniques}

\paragraph{Extractor machinery.} It was shown in \cite{Zuc97} that samplers are equivalent to {\it extractors}, objects that convert weakly random distributions into almost uniform distributions.
Therefore the techniques of constructing extractors are extremely useful in constructing curve samplers.
Our construction employs the technique of block source extraction \cite{NZ96, Zuc97, SZ99}. In addition, we also use the techniques appeared in \cite{GUV09}, especially their constructions of {\it condensers}.

\paragraph{Limited independence.}
It is well known that points on a random degree-$(t-1)$ curve are $t$-wise independent. So we may simply pick a random curve and use tail inequalities to bound the confidence error. However, the sample complexity is too high, and hence we need to use the technique of {\em iterated sampling} to reduce the number of sample points.

\paragraph{Iterated sampling.}
Iterated sampling is a useful technique for constructing explicit randomness-efficient samplers \cite{BR94, TU06}. The idea is first picking a large sample from the domain and then draw a sub-sample from the previous sample.
The drawback of iterated sampling, however, is that it invests randomness twice while the confidence error does not shrink correspondingly. To remedy this problem, we add another ingredient into our construction, namely the technique of {\em error reduction}.

\paragraph{Error reduction via list-recoverable codes.}
We will use explicit list-recoverable codes \cite{GI01}, a strengthening of list-decodable codes. More specifically, we will employ the list-recoverability from (folded) Reed-Solomon codes \cite{GR08, GUV09}.
List-recoverable codes provide a way of obtaining samplers with very small confidence error from those with mildly small confidence error. We refer to this transformation as {\em error reduction}, which plays a key role in our construction.

\subsection{Sketch of the construction}

Our curve sampler is the composition of two samplers which we call the {\em outer sampler} and the {\em inner sampler} respectively. The outer sampler picks manifolds (see \defnref{defn_manifold}) of dimension $O(\log m)$ from the domain $\mathcal{M}=\mathbb{F}_q^m$.
The outer sampler has near-optimal randomness complexity but the sample complexity is large. To fix this problem, we employ the idea of iterated sampling. Namely we regard the manifold picked by the outer sampler as the new domain $\mathcal{M}'$, and then construct an inner sampler picking a curve from $\mathcal{M}'$ with small sample complexity.

The outer sampler is obtained by constructing an extractor and then using the extractor-sampler connection \cite{Zuc97}.
We follow the approach in \cite{NZ96, Zuc97, SZ99}: Given an arbitrary random source with enough min-entropy, we will first use a {\em block source converter} to convert it into a {\em block source}, and then feed it to a {\em block source extractor}.
In addition, we need to construct these components carefully so as to maintain the low-degree-ness.
The way we construct the block source converter is different from those in \cite{NZ96, Zuc97, SZ99} (as they are not in the form of low-degree polynomial maps), and is based on the Reed-Solomon condenser proposed in \cite{GUV09}: To obtain one block, we simply feed the random source and a fresh new seed into the condenser, and let the output be the block.
We show that this indeed gives a block source.
%Then we apply the block source extraction technique with a short seed of length $O(\log q)$.

The inner sampler is constructed using techniques of iterated sampling and error reduction.
We start with the basic curve samplers picking totally random curves, and then
apply the error reduction as well as iterated sampling techniques repeatedly to obtain the desired inner sampler. Either of the two operations improves one parameter while worsening some other one:
Iterated sampling reduces sample complexity but increases the randomness complexity, whereas error reduction reduces the confidence error but increases the sample complexity. Our construction applies the two techniques alternately such that (1) we keep the invariant that the confidence error is always exponentially small in the randomness complexity, and (2) the sample complexity is finally brought down to $q$. 
%We remark that the idea of sandwiching several operations to get the desired parameters without spoiling other ones is reminiscent of Reingold's proof that $\mathsf{SL}=\mathsf{L}$ \cite{Rei08} and Dinur's proof of the PCP theorem \cite{Din07}.

\paragraph{Outline.} The next section contains relevant definitions and some basic facts. Section 3 gives the construction of the outer sampler using block source extraction.
Section 4 introduces the techniques of error reduction and iterated sampling, and then uses them to construct the inner sampler. These components are finally put together in Section 5 and yield the curve sampler construction.

\section{Preliminaries}

We denote the set of numbers $\{1,2,\dots,n\}$ by $[n]$.
Given a prime power $q$, write $\mathbb{F}_q$ for the finite field of size $q$. Write $U_{n,q}$ for the uniform distribution over $\mathbb{F}_q^n$. Logarithms are taken with base 2 unless the base is explicitly specified.

Random variables and distributions are represented by upper-case letters whereas their specific values are represented by lower-case letters. Write $x \leftarrow X$ if $x$ is sampled according to distribution $X$.
The support of a distribution $X$ over set $S$ is $\supp(X)\defeq\{x\in S: \Pr[X=x]>0\}$. The statistical distance between distributions $X,Y$ over set $S$ is defined as $\Delta(X,Y)=\max_{T\subseteq S} |\Pr[X\in T]-\Pr[Y\in T]|$.
We say $X$ is {\em $\epsilon$-close} to $Y$ if $\Delta(X,Y)\leq \epsilon$.

For an event $A$, let $\mathbf{I}[A]$ be the indicator variable that evaluates to $1$ if $A$ occurs and $0$ otherwise. For a random variable $X$ and an event $A$ that occurs with nonzero probability, define the {\em conditional distribution} $X|_A$ by $\Pr[X|_A=x]=\frac{\Pr[(X=x) \wedge A]}{\Pr[A]}$.

%We use forms like $\{t_x: x \in I\}\in S^I$ to denote a collection of elements indexed by $I$ with each element in the set $S$. Alternatively, we view $\{t_x: x \in I\}$ as the function from $I$ to $S$ that maps $x$ to $t_x$. We also slightly abuse the notation and use $\{t_x:x \in I\}$ for an (unordered) multi-set.

\paragraph{Manifolds and curves.}
Let $f:\mathbb{F}_q^d \to \mathbb{F}_q^D$ be a polynomial map. We may view $f$ as $D$ individual polynomials $f_i : \mathbb{F}_q^d \to \mathbb{F}_q$ describing its operation on each output coordinate, i.e., $f(x) = (f_1(x), \dots , f_D(x))$ for all $x\in\mathbb{F}_q^d$.
Such maps are called {\em curves} or {\em manifolds}, depending on the dimension $d$.

\begin{defn}[manifold]\label{defn_manifold}
 A manifold in $\mathbb{F}_q^D$ is a polynomial map $M : \mathbb{F}_q^d \to \mathbb{F}_q^D$ where $M_1,\dots,M_D$ are $d$-variate polynomials over $\mathbb{F}_q$.
 We call $d$ the {\em dimension} of $M$. An 1-dimensional manifold is also called a curve.
 A curve of degree $1$ is also called a {\em line}. The degree of $M$ is $\deg(M)\defeq\max\{\deg(M_1),\dots,\deg(M_D))\}$.
\end{defn}

We need the following lemma, generalizing the one in \cite{TU06}. Its proof is deferred to the appendix.
\newcommand{\lemdegree}{
A manifold $f: \lb\mathbb{F}_{q^D}\rb^n\to\lb\mathbb{F}_{q^D}\rb^m$ of degree $t$, when viewed as a manifold $f: \lb\mathbb{F}_q^D\rb^n\to\lb\mathbb{F}_q^D\rb^m$, also has degree at most $t$.
}
\mylemma{lemdegree}{\lemdegree}

\paragraph{Basic line/curve samplers} The simplest line (resp. curve) samplers are those picking completely random lines (resp. curves), as defined below. We call them {\em basic line} (resp. {\em curve}) {\em samplers}.

\begin{defn}[basic line sampler]\label{defn_line_sampler}
For $m\geq 1$ and prime power $q$, let $\mathsf{Line}_{m,q}:\mathbb{F}_q^{2m}\times\mathbb{F}_q\to\mathbb{F}_q^m$ be the line sampler that picks a completely random line in $\mathbb{F}_q^m$. Formally,
$$
\mathsf{Line}_{m,q}((a,b),y)\defeq (a_1 y+b_1,\dots,a_m y+b_m)
$$
for $a=(a_1,\dots,a_m), b=(b_1,\dots,b_m)\in\mathbb{F}_q^m$ and $y\in\mathbb{F}_q$.
\end{defn}

\begin{defn}[basic curve sampler]
For $m\geq 1$, $t\geq 4$ and prime power $q$, let $\mathsf{Curve}_{m,t,q}:\mathbb{F}_q^{tm}\times\mathbb{F}_q\to\mathbb{F}_q^m$ be the curve sampler that picks a completely random curve of degree $t-1$ in $\mathbb{F}_q^m$. Formally,
$$
\mathsf{Curve}_{m,t,q}((c_0,\dots,c_{t-1}),y)\defeq \lb \sum_{i=0}^{t-1} c_{i,1} y^i,\dots,\sum_{i=0}^{t-1} c_{i,m} y^i\rb
$$
for each $c_0=(c_{0,1},\dots,c_{0,m}),\dots,c_{t-1}=(c_{t-1,1},\dots,c_{t-1,m})\in\mathbb{F}_q^m$ and $y\in\mathbb{F}_q$.
\end{defn}

\begin{rem}
Note that $\mathsf{Line}_{m,q}$ has degree $2$
and $\mathsf{Curve}_{m,t,q}$ has degree $t$ as polynomial maps.
\end{rem}

Tail inequalities for limited independent random variables imply the following lemmas stating that basic line/curve samplers are indeed good samplers.

\newcommand{\lemlinesampler}{
For $\epsilon>0$, $m\geq 1$ and prime power $q$, $\mathsf{Line}_{m,q}$ is an $\lb\epsilon,\frac{1}{\epsilon^2 q}\rb$ line sampler.
}
\mylemma{lemlinesampler}{\lemlinesampler}

\newcommand{\lemcurvesampler}{
For $\epsilon>0$, $m\geq 1$, $t\geq 4$ and sufficiently large prime power $q=(t/\epsilon)^{O(1)}$, $\mathsf{Curve}_{m,t,q}$ is an $\lb\epsilon,q^{-t/4}\rb$ sampler.
}
\mylemma{lemcurvesampler}{\lemcurvesampler}

See the appendix for the proofs of \lemref{lemlinesampler} and \lemref{lemcurvesampler}.

\paragraph{Extractors and condensers.}
A {\em (seeded) extractor} is an object that takes an imperfect random variable called the {\em (weakly) random source}, invests a small amount of randomness called the {\em seed}, and produces an output whose distribution is very close to the uniform distribution.

\begin{defn}[$q$-ary min-entropy]
We say $X$ has {\em $q$-ary min-entropy} $k$ if for any $x\in S$, it holds that $\Pr[X=x]\leq q^{-k}$ (or equivalently, $X$ has min-entropy $k\log q$).
\end{defn}

\begin{defn}[condenser/extractor]
Given a function $f : \mathbb{F}_q^n \times \mathbb{F}_q^d \to \mathbb{F}_q^m$, we say $f$ is an {\em $k_1 \to_{\epsilon,q} k_2$ condenser} if for every distribution $X$ with $q$-ary min-entropy $k_1$, $f(X, U_{d,q})$ is $\epsilon$-close to a distribution with $q$-ary min-entropy $k_2$.
We say $f$ is a {\em  $(k,\epsilon,q)$ extractor} if it is a $k \to_{\epsilon,q} m$ condenser.
\end{defn}

\begin{rem}
We are interested in extractors and samplers that are polynomial maps. For such an object $f$, we denote by $\deg(f)$ its degree as a polynomial map.
\end{rem}

%A random source $X$ is {\em flat} if it is uniformly distributed over its support.
%As a random source $X$ of min-entropy $k$ is a convex combination of flat sources of min-entropy $k$.
%we may assume that the input is always a flat source of min-entropy $k$ when proving the extractor or condenser property for min-entropy threshold $k$.

The following connection between extractors and samplers was observed in \cite{Zuc97}.

\newcommand{\thmequiv}{
Given a map $f:\mathbb{F}_q^n \times \mathbb{F}_q^d \to \mathbb{F}_q^m$, we have the following:
\begin{enumerate}
\item If $f$ is a $(k,\epsilon,q)$ extractor, then it is also an $(\epsilon, \delta)$ sampler where $\delta=2q^{k-n}$.
\item If $f$ is an $(\epsilon/2,\delta)$ sampler where $\delta=\epsilon q^{k-n}$, then it is also a $(k,\epsilon,q)$ extractor.
\end{enumerate}
}
\mythm{equiv}{\thmequiv}{\cite{Zuc97}, restated}

See the appendix for the proof.

\section{Outer sampler}

In this section we construct a sampler whose randomness complexity is optimal up to a constant factor.
We refer to it as the ``outer sampler".
%\subsection{Block source extraction}

We need the machinery of block source extraction.

\begin{defn}[block source \cite{CG88}]
A random source $X=(X_1,\dots, X_s)$ over $\mathbb{F}_q^{n_1}\times\dots\times\mathbb{F}_q^{n_s}$ is a {\em $(k_1,\dots,k_s)$ $q$-ary block source} if for any $i\in [s]$ and $(x_1,\dots,x_{i-1})\in\supp(X_1,\dots,X_{i-1})$, the conditional distribution $X_i|_{X_1=x_1,\dots,X_{i-1}=x_{i-1}}$ has $q$-ary min-entropy $k_i$. Each $X_i$ is called a {\em block}.
\end{defn}

\begin{defn}[block source extractor]
A function $E : (\mathbb{F}_q^{n_1}\times\cdots\times\mathbb{F}_q^{n_s}) \times \mathbb{F}_q^d \to \mathbb{F}_q^m$ is called a {\em $((k_1,\dots,k_s),\epsilon,q)$ block source extractor} if for any $(k_1,\dots,k_s)$ $q$-ary block source $(X_1,\dots, X_s)$ over $\mathbb{F}_q^{n_1}\times\dots\times\mathbb{F}_q^{n_s}$, the distribution $E((X_1,\dots,X_s), U_{d,q})$ is $\epsilon$-close to $U_{m,q}$.
%In addition, we say $E$ has {\em degree} $t$ if $E$ has degree $t$ as a manifold in $\mathbb{F}_q^{n_1+\dots+n_s+d}$.
\end{defn}

One nice property of block sources is that their special structure allows us to compose several extractors and get a block source extractor with only a small amount of randomness invested.

\begin{defn}[block source extraction via composition]\label{def_block_ext_comp}
Let $s\geq 1$ be an integer and $E_i:\mathbb{F}_q^{n_i}\times \mathbb{F}_q^{d_i} \to \mathbb{F}_q^{m_i}$ be a map for each $i\in [s]$.
Suppose that $m_i\geq d_{i-1}$ for all $i\in [s]$, where we set $d_0 = 0$.
Define $E=\mathsf{BlkExt}(E_1,\dots,E_s)$ as follows:
$$
\begin{aligned}
E:(\mathbb{F}_q^{n_1}\times\dots\times\mathbb{F}_q^{n_s})\times\mathbb{F}_q^{d_s} &\to (\mathbb{F}_q^{m_1-d_0}\times\dots\times\mathbb{F}_q^{m_s-d_{s-1}})\\
 ((x_1, \dots , x_s ), y_s ) &\mapsto (z_1, \dots , z_s)
\end{aligned}
$$
where for $i=s, \dots, 1$, we iteratively define $(y_{i-1},z_i)$ to be a partition of $E_i(x_i, y_i)$ into the prefix $y_{i-1}\in\mathbb{F}_q^{d_{i-1}}$ and the suffix $z_i\in\mathbb{F}_q^{m_i-d_{i-1}}$.
\end{defn}

See \autoref{fig_blk_comp} for an illustration of the above definition.

\begin{figure}[hbt]
  \centering
  \begin{tikzpicture}

\node (c1) [component] {$E_s$}; 
\node (p2) [right=of c1, point] {};
\node (c2) [right=of p2, component] {$E_{s-1}$}; 
\node (p3) [right=of c2, point] {};
\node (p4) [right=of p3, txt] {$\cdots$};
\begin{scope}[node distance=1cm]
\node (p4a) [above=of p4, label] {$\cdots$};
\node (p4b) [below=of p4, label] {$\cdots$};
\end{scope}
\node (p5) [right=of p4, point] {};
\node (c3) [right=of p5, component] {$E_2$}; 
\node (p7) [right=of c3, point] {};
\node (c4) [right=of p7, component] {$E_1$};

\draw [pre] (c1.west) -- node [above, very near end] {$Y_s$} +(-1.5cm,0);
\draw [pre] (c1.north) -- node [left, very near end] {$X_s$} +(0,1.5cm);
\draw [pre] (c2.north) -- node [left, very near end] {$X_{s-1}$} +(0,1.5cm);
\draw [pre] (c3.north) -- node [left, very near end] {$X_2$} +(0,1.5cm);
\draw [pre] (c4.north) -- node [left, very near end] {$X_1$} +(0,1.5cm);
\draw [post] (c1.240) -- node [left, very near end] {$Z_s$} +(0,-1.5cm);
\draw [post] (c2.240) -- node [left, very near end] {$Z_{s-1}$} +(0,-1.5cm);
\draw [post] (c3.240) -- node [left, very near end] {$Z_2$} +(0,-1.5cm);
\draw [post] (c4.south) -- node [left, very near end] {$Z_1$} +(0,-1.5cm);

\draw [post] (c1.300) -- +(0,-0.3cm) -| node [below, near start] {$Y_{s-1}$} (p2)  -- (c2.west);
\draw [post] (c2.300) -- +(0,-0.3cm) -| node [below, near start] {$Y_{s-2}$} (p3)  -- (p4.west);
\draw [post] ($(p4.300) + (0,-0.3cm)$) -| node [below, near start] {$Y_2$} (p5) -- (c3.west);
\draw [post] (c3.300) -- +(0,-0.3cm) -| node [below, near start] {$Y_1$} (p7)  -- (c4.west);

\node [component, inner sep=0.8cm, label={[label distance=-0.8cm] 90:$\mathsf{BlkExt}(E_1,\dots,E_s)$}, fit=(c1) (c2)
(c3) (c4)] {};
\end{tikzpicture}
  \caption{The composed block source extractor $\mathsf{BlkExt}(E_1,\dots,E_s)$.}
  \label{fig_blk_comp}
\end{figure}
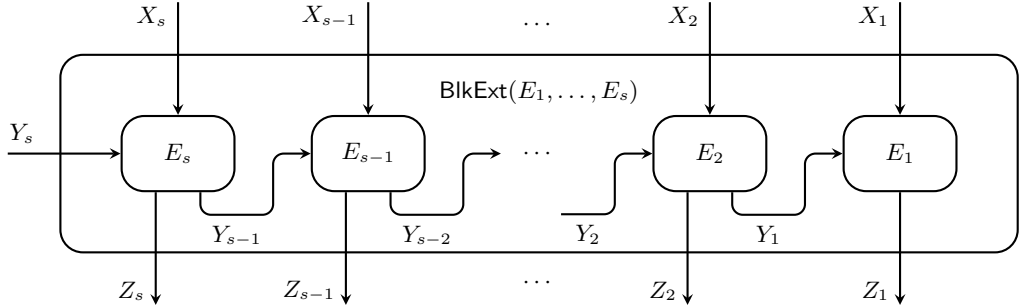

\newcommand{\lemblockextraction}{
Let $s\geq 1$ be an integer and for each $i\in [s]$, let $E_i:\mathbb{F}_q^{n_i}\times \mathbb{F}_q^{d_i} \to \mathbb{F}_q^{m_i}$ be a $(k_i, \epsilon_i, q)$ extractor of degree $t_i\geq 1$.
Then $\mathsf{BlkExt}(E_1,\dots,E_s)$ is a $((k_1,\dots,k_s),\epsilon,q)$ block source extractor of degree $t$ where $\epsilon = \sum_{i=1}^s \epsilon_i$ and $t=\prod_{i=1}^s t_i$.
}
\mylemma{lemblockextraction}{\lemblockextraction}

The proof is deferred to the appendix.

\subsection{Block source conversion}

\begin{defn}[block source converter \cite{NZ96}]
A function $C : \mathbb{F}_q^n \times \mathbb{F}_q^d \to \mathbb{F}_q^{m_1}\times\cdots\times\mathbb{F}_q^{m_s}$ is a {\it $(k,(k_1,\dots,k_s),\epsilon,q)$ block source converter} if for any random source $X$ over $\mathbb{F}_q^n$ with $q$-ary min-entropy $k$, the output $C(X,U_{d,q})$ is $\epsilon$-close to a $(k_1,\dots,k_s)$ $q$-ary block source.
\end{defn}

It was shown in \cite{NZ96} that one can obtain a block by choosing a pseudorandom subset of bits of the random source. Yet the analysis is pretty delicate and cumbersome. Furthermore the resulting extractor does not have a nice algebraic structure. We observe that the following condenser from Reed-Solomon codes in \cite{GUV09} can be used to obtain blocks and is a low-degree manifold.

\begin{defn}[condenser from Reed-Solomon codes \cite{GUV09}]\label{defn_reed_solomon}
Let $\zeta\in\mathbb{F}_q$ be a generator of the multiplicative group $\mathbb{F}_q^\times$. 
Define $\mathsf{RSCon}_{n,m,q}:\mathbb{F}_q^n \times \mathbb{F}_q \to \mathbb{F}_q^{m}$ for $n,m\geq 1$ and prime power $q$:
$$
\mathsf{RSCon}_{n,m,q}(x,y)=\lb y,f_x(y),f_x(\zeta y),\dots, f_x(\zeta^{m-2} y) \rb
$$
where $f_x(Y)=\sum_{i=0}^{n-1} x_i Y^i$ for $x=(x_0,x_1,\dots, x_{n-1})\in\mathbb{F}_q^n$.
\end{defn}

\begin{thm}[\cite{GUV09}]\label{thm-condenser}
$\mathsf{RSCon}_{n,m,q}$ is a
$m \to_{\epsilon,q} 0.99 m $ condenser for large enough $q\geq (n/\epsilon)^{O(1)}$.
\end{thm}

\begin{rem}
The condenser $\mathsf{RSCon}_{n,m,q}(x,y)$ is a degree-$n$ manifold, as each monomial in any of its coordinate is of the form $y$ or $x_i(\zeta^j y)^i$ for some $0\leq i\leq n-1$.
\end{rem}

We apply the above condenser to the source with independent seeds to obtain a block source.

\begin{defn}[block source converter via condensing]\label{defn-blkcont}
For integers $n,m_1,\dots,m_s\geq 1$ and prime power $q$, define the function $\mathsf{BlkCnvt}_{n,(m_1,\dots,m_s),q}:\mathbb{F}_q^n\times\mathbb{F}_q^s\to \mathbb{F}_q^{m_1+\dots+m_s}$ by
$$
\begin{aligned}
\mathsf{BlkCnvt}_{n,(m_1,\dots,m_s),q}(x,y)
=(\mathsf{RSCon}_{n,m_1,q}(x,y_1),\dots,\mathsf{RSCon}_{n,m_s,q}(x,y_s))
\end{aligned}
$$
for $x\in\mathbb{F}_q^n$ and $y=(y_1,\dots,y_s)\in\mathbb{F}_q^s$.
\end{defn}
The function $\mathsf{BlkCnvt}_{n,(m_1,\dots,m_s),q}$ is indeed a block source converter. The intuition is that conditioning on the values of the previous blocks, the random source $X$ still has enough min-entropy, and hence we may apply the condenser to get the next block. Formally, we have the following statement whose proof is deferred to the appendix.

\newcommand{\thmconverter}{
For $\epsilon>0$, integers $s,n,m_1,\dots,m_s\geq 1$ and sufficiently large prime power $q=(n/\epsilon)^{O(1)}$, the function $\mathsf{BlkCnvt}_{n,(m_1,\dots,m_s),q}$ is a $(k,(k_1,\dots,k_s),3s\epsilon,q)$ block source converter of degree $n$ where $k = \sum_{i=1}^s m_i + \log_q(1/\epsilon)$ and each $k_i=0.99 m_i$.
}
\mytheorem{thmconverter}{\thmconverter}

\subsection{Construction of the outer sampler}

By \lemref{lemlinesampler} and \thmref{equiv}, the basic line samplers are also extractors.

\begin{lem}\label{lem_line_extractor}
For $\epsilon>0$, $m\geq 1$ and prime power $q$, $\mathsf{Line}_{m,q}$ is a $(k,\epsilon,q)$ extractor of degree $2$ where $k=2m-1+3\log_q(1/\epsilon)$.
\end{lem}

We employ \lemref{lemblockextraction} and compose the basic line samplers to get a block source extractor. It is then applied to a block source obtained from the block source converter.

\begin{defn}[Outer Sampler]\label{defn_outer_samp}
For $\delta>0$, $m=2^s$ and prime power $q$, let $n=4m+\left\lceil\log_q(2/\delta)\right\rceil$, $d=s+1$, and
$d_i=2^{s-i}$ for $i\in [s]$.
For $i\in [s]$, view $\mathsf{Line}_{2,q^{d_i}}:\mathbb{F}_{q^{d_i}}^4\times\mathbb{F}_{q^{d_i}}\to\mathbb{F}_{q^{d_i}}^2$ as a manifold over $\mathbb{F}_q$: $\mathsf{Line}_{2,q^{d_i}}:\mathbb{F}_q^{4d_i}\times\mathbb{F}_q^{d_i}\to\mathbb{F}_q^{2d_i}$.
%Note that the output of $\mathsf{Line}_{2,q^{d_i}}$ and the seed of $\mathsf{Line}_{2,q^{d_{i-1}}}$ (as extractors) are both in $\mathbb{F}_q^{2d_i}=\mathbb{F}_q^{d_{i-1}}$.
Composing these line samplers $\mathsf{Line}_{2,q^{d_i}}$ for $i\in [s]$ gives the function $\mathsf{BlkExt}(\mathsf{Line}_{2,q^{d_1}},\dots,\mathsf{Line}_{2,q^{d_s}}):\mathbb{F}_q^{4d_1+\dots+4d_s}\times\mathbb{F}_q\to\mathbb{F}_q^m$.
Finally, define $\mathsf{OuterSamp}_{m,\delta,q}:\mathbb{F}_q^n\times\mathbb{F}_q^d\to\mathbb{F}_q^m$:
$$
\mathsf{OuterSamp}_{m,\delta,q}(x,(y,y'))\defeq\mathsf{BlkExt}(\mathsf{Line}_{2,q^{d_1}},\dots,\mathsf{Line}_{2,q^{d_s}})\big(\mathsf{BlkCnvt}_{n,(4d_1,\dots,4d_s),q}(x,y),y'\big)
$$
for $x\in\mathbb{F}_q^n$, $y\in\mathbb{F}_q^s$ and $y'\in\mathbb{F}_q$.
\end{defn}

See \figref{fig_outer} for an illustration of the above construction.

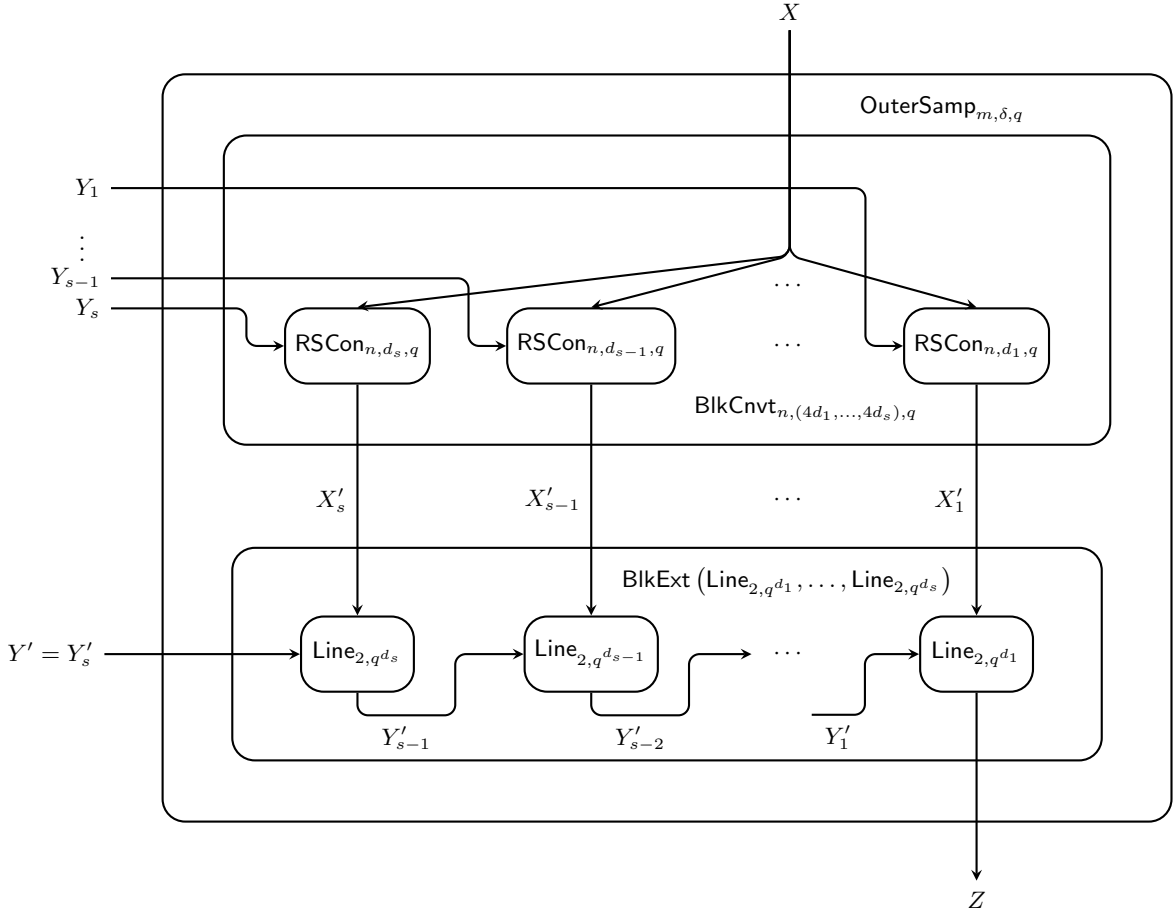
\begin{figure}[hbt]
  \centering
  \begin{tikzpicture}
\node (p1) [point] {};
\node (c1) [right=of p1, component] {$\mathsf{RSCon}_{n,d_s,q}$};
\node (p2) [right=of c1, point] {};
\node (c2) [right=of p2, component] {$\mathsf{RSCon}_{n,d_{s-1},q}$};
\node (p3) [right=of c2, point] {};
\node (p4) [right=of p3, txt] {$\cdots$};
\draw ($(p4)+(0,0.8cm)$) node (p4a) {$\cdots$};
\draw ($(p4)+(0,2cm)$) node [point] (upper) {};
\draw ($(p4)+(0,1.2cm)$) node [point] (input) {};
\node (p5) [right=of p4, point] {};
\node (c4) [right=of p5, component] {$\mathsf{RSCon}_{n,d_1,q}$};

%\draw [ord] (input) -- +(0,3cm) node [anchor=south] {$X$};
\draw [pre] (c1.north) -- (input) -- +(0,3cm) node [anchor=south] {$X$};
\draw [pre] (c2.north) -- (input) -- +(0,3cm);
\draw [pre] (c4.north) -- (input) -- +(0,3cm);

\draw [pre] (c1.west) -- (p1) |- +(-1.8cm,+0.5cm)  node [anchor=east] {$Y_s$};
\draw [pre] (c2.west) -- (p2) |- ($(p1)+(-1.8cm,+0.9cm)$)  node [anchor=east] {$Y_{s-1}$};
\draw ($(p1)+(-2cm,+1.4cm)$) node [anchor=east] {$\vdots$};
\draw [pre] (c4.west) -- (p5) |- ($(p1)+(-1.8cm,+2.1cm)$)  node [anchor=east] {$Y_1$};

\node [component, inner sep=0.8cm, label={[label distance=-0.8cm] 280:$\mathsf{BlkCnvt}_{n,(4d_1,\dots,4d_s),q}$}, fit=(c1) (c2)
(c4) (upper)] (cnvt) {};

\draw ($(c1)+(0,-4.1cm)$) node (tc1) [component] {$\mathsf{Line}_{2,q^{d_s}}$};
\draw ($(c2)+(0,-4.1cm)$) node (tc2) [component] {$\mathsf{Line}_{2,q^{d_{s-1}}}$};
\draw ($(c4)+(0,-4.1cm)$) node (tc4) [component] {$\mathsf{Line}_{2,q^{d_1}}$};
\draw ($(tc1)+(1.3cm,0)$) node (tp2) [point] {};
\draw ($(tc2)+(1.3cm,0)$) node (tp3) [point] {};
\draw ($(p4)+(0,-4.1cm)$) node (tp4) [txt] {$\cdots$};
\draw ($(tp4)+(0,2.05cm)$) node (tp4a) [txt] {$\cdots$};
\draw ($(tp4)+(1cm,0)$) node (tp7) [point] {};

\draw [post] (c1.south) -- node [left, midway] {$X'_s$} (tc1.north);
\draw [post] (c2.south) -- node [left, midway] {$X'_{s-1}$} (tc2.north);
\draw [post] (c4.south) -- node [left, midway] {$X'_1$} (tc4.north);

\draw [pre] (tc1.west) -- +(-2.6cm,0) node [anchor=east] {$Y'=Y'_s$};
\draw [post] (tc4.south) -- +(0,-2.5cm) node [anchor=north] {$Z$};

\draw [post] (tc1.south) -- +(0,-0.3cm) -| node [below, near start] {$Y'_{s-1}$} (tp2)  -- (tc2.west);
\draw [post] (tc2.south) -- +(0,-0.3cm) -| node [below, near start] {$Y'_{s-2}$} (tp3)  -- (tp4.west);
\draw [post] ($(tp4.300) + (0,-0.3cm)$) -| node [below, near start] {$Y'_1$} (tp7) -- (tc4.west);

\node [component, inner sep=0.9cm, label={[xshift=1.6cm ,label distance=-0.8cm] 90:$\mathsf{BlkExt}\left(\mathsf{Line}_{2,q^{d_1}},\dots,\mathsf{Line}_{2,q^{d_s}}\right)$}, fit=(tc1) (tc2) (tc4)] (ext) {};

\node [component, inner sep=0.8cm, label={[label distance=-0.9cm] 60:$\mathsf{OuterSamp}_{m,\delta,q}$}, fit=(cnvt) (ext)] {};
\end{tikzpicture} 
  \caption{The outer sampler $\mathsf{OuterSamp}_{m,\delta,q}$.}
  \label{fig_outer}
\end{figure}

\newcommand{\thmoutersamp}{
For any $\epsilon,\delta>0$, integer $m\geq 1$, and sufficiently large prime power $q\geq(n/\epsilon)^{O(1)}$,
$\mathsf{OuterSamp}_{m,\delta,q}$ is an $(\epsilon,\delta)$ sampler of degree $t$ where $d=O(\log m)$, $n=O\lb m+\log_q(1/\delta)\rb$ and $t=O\lb m^2+m\log_q(1/\delta)\rb$.
}
\mytheorem{thmoutersamp}{\thmoutersamp}

\thmref{thmoutersamp} is proven by combining \lemref{lemdegree}, \thmref{equiv}, \lemref{lemblockextraction}, \thmref{thmconverter}, and \lemref{lem_line_extractor}. We defer the proof to the appendix.

\begin{rem}
We assume $m$ is a power of $2$ above. For general $m$, simply pick $m'=2^{\lceil\log m\rceil}$ and let $\mathsf{OuterSamp}_{m,\delta,q}$ be the composition of $\mathsf{OuterSamp}_{m',\delta,q}$ with the projection $\pi:\mathbb{F}_q^{m'}\to\mathbb{F}_q^m$ onto the first $m$ coordinates. It yields an $(\epsilon,\delta)$ sampler of degree $t$ for $\mathbb{F}_q^m$ since $\pi$ is linear, and approximating the density of a subset $A$ in $\mathbb{F}_q^m$ is equivalent to approximating the density of $\pi^{-1}(A)$ in $\mathbb{F}_q^{m'}$.
\end{rem}

\section{Inner sampler}

The sampler $\mathsf{OuterSamp}_{m,\delta,q}$ has randomness complexity $O\lb m\log q+\log (1/\delta)\rb$ which is optimal up to a constant factor. Yet the sample complexity is large, being $q^{O(\log m)}$. We remedy this problem by composing it with an ``inner sampler'' with small sample complexity. Its construction is based on two techniques called error reduction and iterated sampling.

\subsection{Error reduction}

Condensers are at the core of many extractor constructions \cite{RSW06,TUZ07,GUV09,TU12}. In the language of samplers, the use of condensers can be regarded as an error reduction technique, as we shall see below.

Given a function $f:\mathbb{F}_q^n\times\mathbb{F}_q^d\to\mathbb{F}_q^m$, define $\mathrm{LIST}_f(T,\epsilon)\defeq\left\{x\in\mathbb{F}_q^n: \Pr_y[f(x,y)\in T]>\epsilon\right\}$ for any $T\subseteq\mathbb{F}_q^m$ and $\epsilon>0$.
We are interesting in functions $f$ exhibiting a ``list-recoverability'' property that the size of $\mathrm{LIST}_f(T,\epsilon)$ is kept small when $T$ is not too large.
\begin{defn}
A function $f:\mathbb{F}_q^n\times\mathbb{F}_q\to\mathbb{F}_q^m$ is {\it $(\epsilon, L, H)$ list-recoverable} if $|\mathrm{LIST}_f(T,\epsilon)|\leq H$ for all $T\subseteq\mathbb{F}_q^m$ of size at most $L$.
\end{defn}

We then define an operation $\star$ as follows.

\begin{defn}
For functions $f:\mathbb{F}_q^n\times\mathbb{F}_q^d\to\mathbb{F}_q^m$ and $S:\mathbb{F}_q^m\times\mathbb{F}_q^{d'}\to\mathbb{F}_q^{m'}$, define $S\star f:\mathbb{F}_q^n\times(\mathbb{F}_q^d\times\mathbb{F}_q^{d'})\to\mathbb{F}_q^{m'}$ such that
$(S\star f)(x,(y,y'))\defeq S(f(x,y),y')$.
\end{defn}

See \figref{fig_reduction} for an illustration.
The following lemma states that a sampler with mildly small confidence error, when composed with a list-recoverable function via the $\star$ operation, gives a sampler with very small confidence error.

\begin{lem}\label{lem_error_reduction}
Suppose $f:\mathbb{F}_q^n\times\mathbb{F}_q^d\to\mathbb{F}_q^m$ is $(\epsilon_1,L,H)$ list-recoverable, and
$S:\mathbb{F}_q^m\times\mathbb{F}_q^{d'}\to\mathbb{F}_q^{m'}$ is an $(\epsilon_2,L/q^m)$ sampler.
Then $S\star f$ is an $(\epsilon_1+\epsilon_2,H/q^n)$ sampler.
\end{lem}

\begin{proof}
Let $A$ be an arbitrary subset of $\mathbb{F}_q^{m'}$.
Let $B=\{y\in\mathbb{F}_q^m: |\mu_{S(y)}-\mu(A)|>\epsilon_2\}$. By the sampling property of $S$, we have $|B|\leq (L/q^m)\cdot q^m=L$ and hence $|\mathrm{LIST}_f(B,\epsilon_1)|\leq H$.
Therefore it suffices to show that for any $x\in\mathbb{F}_q^n\setminus \mathrm{LIST}_f(B,\epsilon_1)$, it holds that $|\mu_{(S\star f)(x)}(A)-\mu(A)|\leq\epsilon_1+\epsilon_2$.

Fix $x\in\mathbb{F}_q^n\setminus \mathrm{LIST}_f(B,\epsilon_1)$. We have
$$
\mu_{(S\star f)(x)}(A)=\Pr_{y,y'}[(S\star f)(x,(y,y'))\in A]=\Pr_{y,y'}[S(f(x,y),y')\in A]
=\E_y\left[\mu_{S(f(x,y))}(A)\right].
$$

Therefore
$$
\begin{aligned}
|\mu_{(S\star f)(x)}(A)-\mu(A)|
=\left|\E_y\left[\mu_{S(f(x,y))}(A)\right]-\mu(A)\right|
\leq \E_y |\mu_{S(f(x,y))}(A)-\mu(A)|\\
\leq \Pr_y[f(x,y)\in B] + \epsilon_2\Pr_y[f(x,y)\not\in B]
\leq \epsilon_1+\epsilon_2.
\end{aligned}
$$
To see the last two steps, note that $|\mu_{S(y)}(A)-\mu(A)|\leq\epsilon_2$ for $y\not\in B$ by definition, and $\Pr_y[f(x,y)\in B]\leq \epsilon_1$ since $x\not\in\mathrm{LIST}_f(B,\epsilon_1)$.
\end{proof}

The condenser $\mathsf{RSCon}_{n,m,q}$ (see \defnref{defn_reed_solomon}) enjoys the following list-recoverability property:

\begin{thm}[\cite{GUV09}]\label{thm_reed_solomon_list}
$\mathsf{RSCon}_{n,m,q}$ is $\lb\epsilon, q^{0.99 m}, q^m\rb$ list-recoverable for sufficiently large $q\geq (n/\epsilon)^{O(1)}$.
\end{thm}

\begin{cor}\label{cor_reed_solomon_reduction}
For any $n\geq m\geq 1$, $\epsilon,\epsilon'>0$ and sufficiently large prime power $q=(n/\epsilon)^{O(1)}$, suppose $S:\mathbb{F}_q^m\times\mathbb{F}_q^d\to\mathbb{F}_q^{m'}$ is an $(\epsilon',q^{-0.01 m})$ sampler of degree $t$,
then $S\star \mathsf{RSCon}_{n,m,q}$ is an $(\epsilon+\epsilon', q^{m-n})$ sampler of degree $nt$.
\end{cor}
\begin{proof}
Apply \lemref{lem_error_reduction} and \thmref{thm_reed_solomon_list}. Note that $\mathsf{RSCon}_{n,m,q}$ has degree $n$. Therefore $\lb S\star \mathsf{RSCon}_{n,m,q}\rb (X,(Y,Y'))=S(\mathsf{RSCon}_{n,m,q}(X,Y),Y')$ has degree $nt$ in its variables $X,Y,Y'$.
\end{proof}

\subsection{Iterated sampling}

%Iterated sampling is a useful technique for constructing randomness-efficient samplers \cite{BR94, TU06}. The idea is first picking a huge sample from the domain. This can be easily done, even in the randomness-optimal way since we are allowed to pick many sample points. For example, we may simply regard the domain $[N]$ as a $2$-dimensional vector space and then pick a random line or curve that contains $\sqrt{N}$ points. To reduce the sample complexity, we just draw a sub-sample from the previous sample. Note that at this sub-sampling step we are dealing with a smaller domain that is a sample picked from the original domain. And the advantage is that the random bits costed for sampling this smaller domain are fewer than those used for the original domain. We may further exploit this idea by iteratively drawing fewer and fewer sample points, each from the previous sample.

%To formalize this idea,
We introduce the operation $\circ$ denoting the composition of two samplers.

\begin{defn} (composed sampler). Given functions $S_1 : \mathbb{F}_q^{n_1}\times\mathbb{F}_q^{d_1}\to \mathbb{F}_q^{d_0}$ and $S_2 : \mathbb{F}_q^{n_2} \times \mathbb{F}_q^{d_2}\to \mathbb{F}_q^{d_1}$, define $S_1\circ S_2:(\mathbb{F}_q^{n_1}\times \mathbb{F}_q^{n_2})\times \mathbb{F}_q^{d_2}\to\mathbb{F}_q^{d_0}$ such that
$(S_1\circ S_2)((x_1,x_2),y)\defeq S_1(x_1,S_2(x_2,y))$.
%for all $x_1\in\mathbb{F}_q^{n_1}$, $x_2\in\mathbb{F}_q^{n_2}$ and $y\in\mathbb{F}_q^{d_2}$.
\end{defn}

See \figref{fig_iterated} for an illustration. The composed sampler $S_1\circ S_2$ first uses its randomness $x_1$ to get the sample $S_1(x_1)=\{S_1(x_1,y):y\in\mathbb{F}_q\}$, and then uses its randomness $x_2$ to get the subsample $\{S_1(x_1,S_2(x_2,y)):y\in\mathbb{F}_q\}\subseteq S_1(x_1)$.
Intuitively, if $S_1$ and $S_2$ are good samplers then so is $S_1\circ S_2$. This is indeed shown by \cite{BR94, TU06} and we formalize it as follows:

\newcommand{\lemresampling}{
Let $S_1 : \mathbb{F}_q^{n_1}\times\mathbb{F}_q^{d_1}\to \mathbb{F}_q^{d_0}$ be an $(\epsilon_1,\delta_1)$ sampler of degree $t_1$.
And let $S_2 : \mathbb{F}_q^{n_2} \times \mathbb{F}_q^{d_2}\to \mathbb{F}_q^{d_1}$ be an $(\epsilon_2,\delta_2)$ sampler of degree $t_2$.
Then $S_1 \circ S_2 :(\mathbb{F}_q^{n_1}\times \mathbb{F}_q^{n_2})\times \mathbb{F}_q^{d_2}\to\mathbb{F}_q^{d_0}$ is an $(\epsilon_1+\epsilon_2, \delta_1+\delta_2)$ sampler of degree $t_1 t_2$.
}
\mylem{lemresampling}{\lemresampling}{\cite{BR94, TU06}}

See the appendix for its proof.

\begin{figure}[hbt]
  \centering
  \begin{minipage}[b]{.48\textwidth}
  \centering
  \begin{tikzpicture}
\node (c1) [component] {$f$}; 
\node (c2) [below=of c1, component] {$S$};

\draw [pre] (c1.north) -- +(0,1cm) node [left, very near end] {$X$};
\draw [pre] (c1.west) -- +(-1cm,0) node [above, very near end] {$Y$};
\draw [pre] (c2.west) -- +(-1cm,0) node [above, very near end] {$Y'$};
\draw [post] (c2.south) -- +(0,-1cm) node [left, very near end] {$Z$};
%\draw ($(c1)+(0.6cm,1cm)$) node (title) {$S\star f$};

\draw [post] (c1) to (c2);
\node [component, inner sep=0.3cm, fit=(c1) (c2)] {};
\end{tikzpicture}
  \caption{The operation $S\star f$.}
  \label{fig_reduction}
  \end{minipage}%
  \begin{minipage}[b]{.48\textwidth}
  \centering
  \begin{tikzpicture}
\node (c1) [component] {$S_2$}; 
\node (p2) [right=of c1, point] {};
\node (c2) [right=of p2, component] {$S_1$}; 
\draw [pre] (c1.west) -- node [above, very near end] {$Y$} +(-1cm,0);
\draw [pre] (c1.north) -- node [left, very near end] {$X_2$} +(0,1cm);
\draw [pre] (c2.north) -- node [left, very near end] {$X_1$} +(0,1cm);
\draw [post] (c2.south) -- node [left, very near end] {$Z$} +(0,-1cm);
\draw [post] (c1.south) -- +(0,-0.2cm) -| node [right, midway] {$Y'$} (p2)  -- (c2.west);
\node [component, inner sep=0.5cm, fit=(c1) (c2)] {};
\end{tikzpicture}
  \caption{The operation $S_1\circ S_2$.}
  \label{fig_iterated}
  \end{minipage}
\end{figure}

\subsection{Construction of the inner sampler}

We use the basic curve samplers as the building blocks and apply the error reduction as well as iterated sampling repeatedly to obtain the inner sampler. The formal construction is as follows.

\begin{defn}[inner sampler]\label{defn_inner_samp}
For $m\geq 1$, $\delta>0$ and prime power $q$, pick $s=\lceil\log m\rceil$ and let $d_i=2^{s-i}$ for $0\leq i\leq s$. Let $n_i=16^i$ for $0\leq i\leq s-1$ and $n_s=16^s+ 20\left\lceil\log_q(1/\delta)\right\rceil$.
Define $S_i:\mathbb{F}_q^{n_i d_i}\times\mathbb{F}_q^{d_i}\to\mathbb{F}_q^m$ for $0\leq i\leq s$ as follows:
\begin{itemize}
\item $S_0:\mathbb{F}_q\times\mathbb{F}_q^{d_0}\to\mathbb{F}_q^m$ projects $(x,y)$ onto the first $m$ coordinates of $y$.
\item $S_i\defeq\lb S_{i-1}\star \mathsf{RSCon}_{\frac{n_i}{4},2n_{i-1},q^{d_i}}\rb \circ \mathsf{Curve}_{3,\frac{n_i}{4},q^{d_i}}$ for $i=1,\dots,s$.
\end{itemize}
Finally, let $\mathsf{InnerSamp}_{m,\delta,q}\defeq S_s$.
\end{defn}

See \figref{fig_inner} for an illustration of the above construction.

\begin{figure}[hbt]
  \centering
  \begin{tikzpicture}

\node (c2) [component] {$S_{i-1}$}; 
\begin{scope}[node distance=1.2cm]
\node (c4) [above=of c2, component] {$\mathsf{RSCon}_{\frac{n_i}{4},2n_{i-1},q^{d_i}}$};
\node (c1) [left=of c4, component] {$\mathsf{Curve}_{3,\frac{n_i}{4},q^{d_i}}$}; 
\end{scope}
\node (p4) [left=of c4, point] {};
\begin{scope}[node distance=1.5cm]
\node (l1) [left=of c1, point] {};
\end{scope}

\draw [pre] (c1.west) -- (l1) node [above, very near end] {$Y_i$};
\draw [pre] (c1.north) -- +(0,1.2cm) node [left, very near end] {$X_{i,1}$};
\draw [pre] (c4.north) -- +(0,1.2cm) node [left, very near end] {$X_{i,2}$};
\draw [post] (c4) to node [right, midway] {$X_{i-1}$} (c2);
\draw [post] (c2.south) -- node [left, very near end] {$Z$} +(0,-1.5cm);
\draw [post] (c1.240) |- node [above, very near end] {$Y_{i-1}$} (c2.west);
\draw [post] (c1.300) -- +(0,-0.3cm) -| (p4)  -- (c4.west);

\node [component, inner sep=0.7cm, label={[label distance=-0.8cm] 270:$S_i$}, fit=(c1) (c2)  (c4)] {};
\end{tikzpicture}
  \caption{The recursive construction of $S_i=\lb S_{i-1}\star \mathsf{RSCon}_{\frac{n_i}{4},2n_{i-1},q^{d_i}}\rb \circ \mathsf{Curve}_{3,\frac{n_i}{4},q^{d_i}}$ used to define the inner sampler. Here $X_i=(X_{i,1},X_{i,2})$ (resp. $X_{i-1}$) and $Y_i$ (resp. $Y_{i-1}$) are the two arguments of $S_i$ (resp. $S_{i-1}$). And $Z$ is the common output of $S_i$ and $S_{i-1}$.}
  \label{fig_inner}
\end{figure}
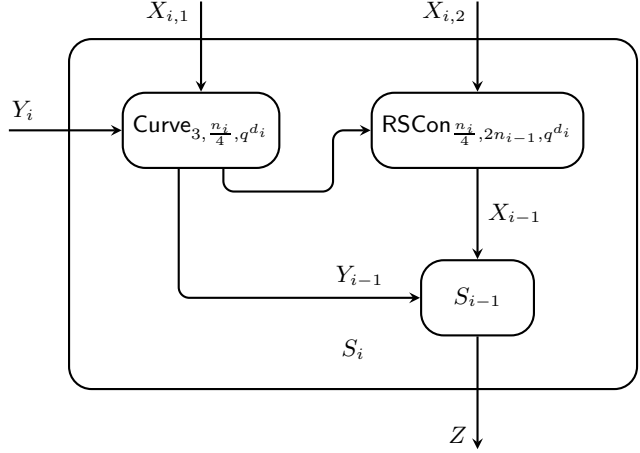

\begin{thm}\label{thm_inner_samp}
For any $\epsilon,\delta>0$, integer $m\geq 1$ and large enough prime power $q\geq\lb\frac{m \log(1/\delta)}{\epsilon}\rb^{O(1)}$, $\mathsf{InnerSamp}_{m,\delta,q}:\mathbb{F}_q^n\times\mathbb{F}_q\to\mathbb{F}_q^m$ is an $(\epsilon, \delta)$ sampler of degree $t$ where $n=O\lb m^{O(1)}+\log_q(1/\delta)\rb$ and $t=O\lb m^{O(\log m)}\log_q^2 (1/\delta)\rb$.
\end{thm}

\begin{proof}
Let $\epsilon'=\frac{\epsilon}{2s}$ and $d_i,n_i,S_i$ be as in \defnref{defn_inner_samp} for $0\leq i\leq s$.
We will prove that each $S_i$ is an $(\epsilon_i,\delta_i)$ sampler of degree $t_i$ where
$\epsilon_i=2i \epsilon'$,
$\delta_i=q^{-n_id_i/20}$,
and $t_i=\prod_{j=1}^i \lb\frac{n_j}{4}\rb^2$.
The theorem follows by noting that $\mathsf{InnerSamp}_{m,\delta,q}=S_s$, $\epsilon=\epsilon_s$, $\delta\geq\delta_s$, and $t=t_s$.

Induct on $i$. The case $i=0$ is trivial.  Consider the case $i>0$ and assume the claim holds for all $i'<i$.
By the induction hypothesis, $S_{i-1}$ is an $(\epsilon_{i-1},\delta_{i-1})$ sampler of degree $t_{i-1}$.

Note that $\delta_{i-1}=q^{-n_{i-1}d_{i-1}/20}\leq q^{-0.01n_{i-1}d_{i-1}}$.
By \corref{cor_reed_solomon_reduction},
$S_{i-1}\star \mathsf{RSCon}_{\frac{n_i}{4},2n_{i-1},q^{d_i}}$
is an $\lb \epsilon_{i-1}+\epsilon', q^{n_{i-1}d_{i-1}-(n_i/4) d_i} \rb$ sampler of degree $\frac{n_i}{4}\cdot t_{i-1}$.
By \lemref{lemcurvesampler}, $\mathsf{Curve}_{3,\frac{n_i}{4},q^{d_i}}$ is an  $\lb\epsilon',q^{-n_i d_i/16}\rb$ sampler of degree $\frac{n_i}{4}$.
Finally by \lemref{lemresampling}, the function
$$
S_i=\lb S_{i-1}\star \mathsf{RSCon}_{\frac{n_i}{4},2n_{i-1},q^{d_i}}\rb \circ \mathsf{Curve}_{3,\frac{n_i}{4},q^{d_i}}
$$
is an $\lb\epsilon_{i-1}+2\epsilon', q^{n_{i-1}d_{i-1}-(n_i/4) d_i}+q^{-n_i d_i/16}\rb$ sampler of degree $\lb\frac{n_i}{4}\rb^2\cdot t_{i-1}$. It remains to check that
$$
\epsilon_i=\epsilon_{i-1}+2\epsilon',\quad \delta_i\geq q^{n_{i-1}d_{i-1}-(n_i/4) d_i}+q^{-n_i d_i/16}\quad \mathrm{and}\quad  t_i=\lb\frac{n_i}{4}\rb^2\cdot t_{i-1}.
$$
which hold by the choices of parameters.
\end{proof}

\section{Putting it together}

We compose the outer sampler and the inner sampler to get the desired curve sampler.

\begin{defn}\label{defn_sampler}
For $m\geq 1$, $\delta>0$ and prime power $q$, define
$$
\mathsf{Samp}_{m,\delta,q}
\defeq\mathsf{OuterSamp}_{m,\delta/2,q}\circ\mathsf{InnerSamp}_{d,\delta/2,q}.
$$
\end{defn}

\begin{thm}[\thmref{thm_sampler_1} restated]\label{thm_sampler}
For any $\epsilon,\delta>0$, integer $m\geq 1$ and sufficiently large prime power $q\geq\lb\frac{m \log(1/\delta)}{\epsilon}\rb^{O(1)}$, the function $\mathsf{Samp}_{m,\delta,q}:\mathbb{F}_q^n\times\mathbb{F}_q\to\mathbb{F}_q^m$ is an $(\epsilon,\delta)$ sampler of degree $t$ where $n=O\lb m+\log_q(1/\delta)\rb$ and $t=\lb m\log_q(1/\delta)\rb^{O(1)}$. In particular, $\mathsf{Samp}_{m,\delta,q}$ is an $(\epsilon,\delta)$ degree-$t$ curve sampler.
\end{thm}

\begin{proof}
By \thmref{thmoutersamp},
$\mathsf{OuterSamp}_{m,\delta/2,q}:\mathbb{F}_q^{n_1}\times\mathbb{F}_q^d\to\mathbb{F}_q^m$ is an $(\epsilon/2,\delta/2)$ sampler of degree $t_1$ where $d=O(\log m)$, $n_1=O\lb m+\log_q(1/\delta)\rb$ and $t_1=O\lb m^2+m\log_q(1/\delta)\rb$.

By \thmref{thm_inner_samp}, $\mathsf{InnerSamp}_{d,\delta/2,q}:\mathbb{F}_q^{n_2}\times\mathbb{F}_q\to\mathbb{F}_q^d$ is an $(\epsilon/2, \delta/2)$ sampler of degree $t_2$ where $n_2=O\lb\lb\log m\rb^{O(1)}+\log_q(1/\delta)\rb$ and $t_2=O\lb \lb\log m\rb^{O(\log\log m)}\log_q^2 (1/\delta)\rb$.

Finally, \lemref{lemresampling} implies that
$\mathsf{Samp}_{m,\delta,q}$
is an $(\epsilon,\delta)$ sampler of degree $t$ with $n=n_1+n_2=O\lb m+\log_q(1/\delta)\rb$ and $t=t_1t_2=\lb m\log_q(1/\delta)\rb^{O(1)}$.
A fortiori, it is a degree-$t$ curve sampler since
the degree of $\mathsf{Samp}_{m,\delta,q}(x,\cdot)$ is bounded by the degree of $\mathsf{Samp}_{m,\delta,q}$ for all $x\in\mathbb{F}_q^n$.
\end{proof}

\section*{Acknowledgements}

The author is grateful to Chris Umans for his support and many helpful discussions.

%\newpage

\bibliographystyle{alpha}
\bibliography{random}

\appendix

\section{Tail probability bounds}

The following bound follows from Chebyshev's inequality:
\begin{lem}\label{lem_pairwise}
Suppose $X_1,\dots,X_n$ are pairwise independent random variables. Let $X=\sum_{i=1}^n X_i$ and $\mu=\E[X]$, and let $A>0$. Then
$$
\Pr[|X-\mu|\geq A]\leq\frac{\sum_{i=1}^n\mathbf{Var}[X_i]}{A^2}.
$$
\end{lem}

We will also use the following tail bound for $t$-wise independent random variables:
\begin{lem}[\cite{BR94}]\label{lem_twise}
Let $t\geq 4$ be an even integer. Suppose $X_1,\dots,X_n$ are $t$-wise independent random variables over $[0,1]$. Let $X=\sum_{i=1}^n X_i$ and $\mu=\E[X]$, and let $A>0$. Then
$$
\Pr[|X-\mu|\geq A]=O\lb\lb\frac{t\mu+t^2}{A^2}\rb^{t/2}\rb.
$$
\end{lem}

\section{Lower bounds}\label{deglb}

We will use the following optimal lower bound for extractors:
\begin{thm}[\cite{RT00}, restated]\label{thm_ext_lower_bound}
Let $E:\mathbb{F}_q^n\times\mathbb{F}_q^d\to\mathbb{F}_q^m$ be a $(k,\epsilon,q)$ extractor. Then
\begin{enumerate}[(a)]
\item if $\epsilon<1/2$ and $q^d\leq q^m/2$, then $q^d=\Omega\lb\frac{(n-k)\log q}{\epsilon^2}\rb$, and
\item if $q^d\leq q^m/4$, then $q^{d+k-m}=\Omega(1/\epsilon^2)$.
\end{enumerate}
\end{thm}

\begin{thm}\label{thm_samp_lower_bound}
Let $S:\mathbb{F}_q^n\times\mathbb{F}_q\to\mathbb{F}_q^m$ be an $(\epsilon,\delta)$ curve sampler where $\epsilon<1/2$ and $m\geq 2$. Then
\begin{enumerate}[(a)]
\item the sample complexity $q=\Omega\lb\frac{\log(2\epsilon/\delta)}{\epsilon^2}\rb$, and
\item the randomness complexity $n\log q\geq (m-1)\log q+\log(1/\epsilon)+\log(1/\delta)-O(1)$.
\end{enumerate}
\end{thm}
\begin{proof}
By \thmref{equiv}, $S$ is a $(k,2\epsilon,q)$ extractor where $k=n-\log_q(2\epsilon/\delta)$. The first claim then follows from \thmref{thm_ext_lower_bound} (a).
Applying \thmref{thm_ext_lower_bound} (b), we get $(1+k-m)\log q\geq \Omega(\log(1/\epsilon))-O(1)$. Therefore
$$
n\log q=k\log q+\log(2\epsilon/\delta)\geq (m-1)\log q+\log(1/\delta)+\Omega(\log(1/\epsilon))-O(1).
$$
\end{proof}

In particular, as $\log(1/\epsilon)=O(\log q)$, the randomness complexity $n\log q$ is at least $\Omega(\log N+\log(1/\delta))$ when the domain size $N=q^m\geq N_0$ for some constant $N_0$. Therefore the randomness complexity in \thmref{thm_sampler_1} is optimal up to a constant factor.

We also present the following lower bound on the degree of curves sampled by a curve sampler:
\begin{thm}\label{thm_deg_lower_bound}
Let $S:\mathcal{N}\times\mathbb{F}_q\to\mathbb{F}_q^m$ be an $(\epsilon,\delta)$ degree-$t$ curve sampler where $m\geq 2$, $\epsilon<1/2$ and $\delta<1$. Then $t=\Omega\lb\log_q(1/\delta)+1\rb$.
\end{thm}
\begin{proof}
Clearly $t\geq 1$.
Suppose $S=(S_1,\dots,S_m)$ and define $S'=(S_1,S_2)$.
Let $\mathcal{C}$ be the set of curves of degree at most $t$ in $\mathbb{F}_q^2$. Then $|\mathcal{C}|=q^{2(t+1)}$. Consider the map $\tau:\mathcal{N}\to\mathcal{C}$ that sends $x$ to $S'(x,\cdot)$. We can pick $k=\left\lfloor q/2\right\rfloor$ curves $C_1,\dots,C_k\in\mathcal{C}$ such that the union of their preimages
$$
B\defeq\bigcup_{i=1}^k \tau^{-1}(C_i)=\bigcup_{i=1}^k \{x: S'(x,\cdot)=C_i\}
$$
has size at least
$\frac{k|\mathcal{N}|}{|\mathcal{C}|}=\frac{k|\mathcal{N}|}{q^{2(t+1)}}$.

Define $A\subseteq\mathbb{F}_q^m$ by
$$
A\defeq\{C_i(y):i\in[k],~y\in\mathbb{F}_q\}\times\mathbb{F}_q^{m-2},
$$
i.e., let $A$ be the set of points in $\mathbb{F}_q^m$ whose first two coordinates are on at least one curve $C_i$. We have $|A|\leq kq^{m-1}$ and hence $\mu(A)\leq k/q\leq 1/2<1-\epsilon$. On the other hand, it follows from the definition of $A$ that we have $S(x,y)\in A$ for all $x\in B$ and $y\in\mathbb{F}_q$. So $\mu_{S(x)}(A)=1$ for all $x\in B$. Then $\delta\geq \Pr\left[|\mu_{S(x)}(A)-\mu(A)|>\epsilon\right]\geq \frac{|B|}{|\mathcal{N}|}\geq\frac{k}{q^{2(t+1)}}$ and hence
$t\geq\max\left\{1, \frac{1}{2}\log_q(k/\delta)-1\right\}=\Omega\lb\log_q(1/\delta)+1\rb$.
\end{proof}

We remark that the condition $m\geq 2$ is necessary in \thmref{thm_deg_lower_bound} since when $m=1$, the sampler $S$ with $S(x,y)=y$ for all $x\in\mathcal{N}$ and $y\in\mathbb{F}_q$ is a $(0,0)$ degree-$1$ curve sampler.

\section{Omitted Proofs}

\againlemma{lemdegree}{\lemdegree}

\begin{proof}
Write $f=(f_1,\dots,f_m)$. By symmetry we just show $\deg(f_1)\leq t$ as a polynomial map over $\mathbb{F}_q$.
Suppose
$$
f_1(x_1,\dots,x_n)=\sum_{\substack{d=(d_1,\dots,d_n)\\ \sum d_i\leq t}} c_d\prod_{i=1}^n {x_i}^{d_i}.
$$
Let $(e_1,\dots,e_D)$ be the standard basis of $\mathbb{F}_q^D$ over $\mathbb{F}_q$.
Writing the $i$-th variable $x_i\in \mathbb{F}_{q^D}$ as $\sum_{j=1}^D x_{i,j} e_j\in\mathbb{F}_q^D$ with
$x_{i,j}\in \mathbb{F}_q$, and each coefficient $c_d\in F_{q^D}$ as $\sum_{j=1}^D c_{d,j} e_j$ with $c_{d,j}\in \mathbb{F}_q$, we obtain
$$
f_1(x_1,\dots,x_n)=\sum_{\substack{d=(d_1,\dots,d_n)\\ \sum d_i\leq t}} \left(\sum_{j=1}^D c_{d,j} e_j\right) \prod_{i=1}^n \left(\sum_{j=1}^D x_{i,j} e_j\right)^{d_i}.
$$
After multiplying out, each monomial has degree $\max_d \sum_i d_i\leq t$, and their coefficients are polynomials in
the $e_i$ elements. Rewriting each of these values in
the basis $(e_1, \dots, e_D)$ and gathering the coefficients on $e_i$,
we obtain the $i$-th coordinate function of $f_1$ that has degree at most $t$ for all $1\leq i\leq D$. Therefore $f_1: \lb\mathbb{F}_q^D\rb^n\to \mathbb{F}_q^D$ is a manifold of degree at most $t$, and so is $f: \lb\mathbb{F}_q^D\rb^n\to\lb\mathbb{F}_q^D\rb^m$.
\end{proof}

\againlemma{lemlinesampler}{\lemlinesampler}

\begin{proof}
Let $A$ be an arbitrary subset of $\mathbb{F}_q^m$.
Note that $\mathsf{Line}_{m,q}$ picks a line uniformly at random.
It is well known that the random variables $\mathsf{Line}_{m,q}(U_{2m,q},y)$ with $y$ ranging over $\mathbb{F}_q$ are pairwise independent. So the indicator variables $\mathbf{I}[\mathsf{Line}_{m,q}(U_{2m,q},y)\in A]$ with $y$ ranging over $\mathbb{F}_q$ are also pairwise independent. Applying \lemref{lem_pairwise}, we get
$$
\begin{aligned}
&\Pr_{x\leftarrow U_{2m,q}}\left[\left|\mu_{\mathsf{Line}_{m,q}(x)}(A)-\mu(A)\right|>\epsilon\right]\\
&=\Pr\left[\left|\sum_{y\in\mathbb{F}_q}\mathbf{I}[\mathsf{Line}_{m,q}(U_{2m,q},y)\in A]-\E\left[\sum_{y\in\mathbb{F}_q}\mathbf{I}[\mathsf{Line}_{m,q}(U_{2m,q},y)\in A]\right]\right|>\epsilon q\right]\\
&\leq\frac{\sum_{y\in\mathbb{F}_q}\mathbf{Var}\big[\mathbf{I}[\mathsf{Line}_{m,q}(U_{2m,q},y)\in A]\big]}{\epsilon^2q^2}\\
&\leq\frac{1}{\epsilon^2 q}.
\end{aligned}
$$
By definition, $\mathsf{Line}_{m,q}$ is an $\lb\epsilon,\frac{1}{\epsilon^2 q}\rb$ line sampler.
\end{proof}

\againlemma{lemcurvesampler}{\lemcurvesampler}

\begin{proof}
Let $A$ be an arbitrary subset of $\mathbb{F}_q^m$.
Note that $\mathsf{Curve}_{m,t,q}$ picks a degree-$(t-1)$ curve uniformly at random.
It is well known that the random variables $\mathsf{Curve}_{m,t,q}(U_{tm,q},y)$ with $y$ ranging over $\mathbb{F}_q$ are $t$-wise independent. So the indicator variables $\mathbf{I}[\mathsf{Curve}_{m,t,q}(U_{tm,q},y)\in A]$ with $y$ ranging over $\mathbb{F}_q$ are also $t$-wise independent. Applying \lemref{lem_twise}, we get
$$
\begin{aligned}
&\Pr_{x\leftarrow U_{tm,q}}\left[\left|\mu_{\mathsf{Curve}_{m,t,q}(x)}(A)-\mu(A)\right|> \epsilon\right]\\
&=\Pr\left[\left|\sum_{y\in\mathbb{F}_q}\mathbf{I}[\mathsf{Curve}_{m,t,q}(U_{tm,q},y)\in A]-\E\left[\sum_{y\in\mathbb{F}_q}\mathbf{I}[\mathsf{Curve}_{m,t,q}(U_{tm,q},y)\in A]\right]\right|> \epsilon q\right]\\
&=O\lb\lb\frac{tq\mu(A)+t^2}{\epsilon^2 q^2}\rb^{t/2}\rb
\leq q^{-t/4}
\end{aligned}
$$
provided that $q\geq (t/\epsilon)^{O(1)}$ is sufficiently large.
By definition, $\mathsf{Curve}_{m,t,q}$ is an $\lb\epsilon,q^{-t/4}\rb$ sampler.
\end{proof}

\againtheorem{equiv}{\thmequiv}

\begin{proof}
(Extractor to sampler) Assume to the contrary that $f$ is not an $(\epsilon, \delta)$ sampler. Then there exists a subset $A\subseteq \mathbb{F}_q^m$ with $\Pr_x[ |\mu_{f(x)}(A) - \mu(A)| > \epsilon ] > \delta$. Then either $\Pr_x[ \mu_{f(x)}(A) - \mu(A) > \epsilon ]>\delta/2$ or $\Pr_x[ \mu_{f(x)}(A) - \mu(A) < -\epsilon ]>\delta/2$.
Assume $\Pr_x[ \mu_{f(x)}(A) - \mu(A) > \epsilon ]>\delta/2$ (the other case is symmetric). Let $X$ be the uniform distribution over the set of $x$ such that $\mu_{f(x)}(A) - \mu(A) > \epsilon$, i.e., $\Pr_y[f(x,y)\in A] - \Pr_x[x\in A] > \epsilon$. Then $|f(X,U_{d,q})-U_{m,q}|> \epsilon$. But the $q$-ary min-entropy of $X$ is at least $\log_q((\delta/2)q^n)\geq k$, contradicting the extractor property of $f$.

(Sampler to extractor)  Assume to the contrary that $f$ is not a $(k,\epsilon,q)$ extractor. Then there exists a subset $A\subseteq \mathbb{F}_q^m$ and a random source $X$ of $q$-ary min-entropy $k$ satisfying the property that $|\Pr[f(X,U_{d,q})\in A]-\mu(A)|>\epsilon$. We may assume $X$ is a flat source (i.e. uniformly distributed over its support) since a general source with $q$-ary min-entropy $k$ is a convex combination of flat sources with $q$-ary min-entropy $k$.
Note that $|\supp(X)|\geq q^k$.
By the averaging argument, for at least an $\epsilon$-fraction of $x\in \supp(X)$, we have $|\Pr[f(x,U_{d,q})\in A]-\mu(A)|>\epsilon/2$. But it implies that for $x$ uniformly chosen from $\mathbb{F}_q^n$, with probability at least $\epsilon\mu(\supp(X))=\delta$ we have $|\mu_{f(x)}(A)-\mu(A)|>\epsilon/2$, contradicting the sampling property of $f$.
\end{proof}

\againlemma{lemblockextraction}{\lemblockextraction}

\begin{proof}
Induct on $s$. When $s=1$ the claim follows from the extractor property of $E_1$.

When $s>1$, assume the claim holds for all $s'<s$.
Let $E=\mathsf{BlkExt}(E_1,\dots,E_s)$ and $E'=\mathsf{BlkExt}(E_2,\dots,E_s)$.
By the induction hypothesis, $E'$ is a $((k_2,\dots,k_s),\epsilon',q)$ block source extractor of degree $t'$ where $\epsilon'=\sum_{i=2}^s \epsilon_i$ and $t'=\prod_{i=2}^s t_i$.
Let $(X_1,\dots,X_s)$ be an arbitrary $(k_1,\dots,k_s)$ $q$-ary block source over $\mathbb{F}_q^{n_1}\times\dots\times\mathbb{F}_q^{n_s}$.
Let $(Y_{i-1},Z_i)$ be the output of $E_i$ and $(Z_1,\dots,Z_s)$ be the output of $E$ when $(X_1,\dots,X_s)$ is fed to $E$ as the input and an independent uniform distribution $Y_s=U_{d_s,q}$ is used as the seed. The output of $E'$ is then $(Y_1, Z_2, \dots , Z_s)$.

Fix $x\in\supp(X_1)$.
By the definition of block sources, the distribution $(X_2,\dots,X_s)|_{X_1=x}$ is a $(k_2,\dots,k_s)$ $q$-ary block source. Also note that $Y_s|_{X_1=s}=Y_s$ is uniform distributed and independent from $X_2,\dots, X_s$.
By the induction hypothesis, the distribution
$$
(Y_1, Z_2, \dots , Z_s)|_{X_1=x}=E'((X_2,\dots, X_s), Y_s)|_{X_1=x}
$$
is $\epsilon'$-close to $(U_{d_1,q},U_{m_2-d_1,q},\dots,U_{m_s-d_{s-1},q})$.
As this holds for all $x\in\supp(X_1)$, the distribution $(X_1,Y_1,Z_2,\dots,Z_s)$ is $\epsilon'$-close to $(X_1,U_{d_1,q},U_{m_2-d_1,q},\dots,U_{m_s-d_{s-1},q})$.
So the distribution
$$
E((X_1,\dots,X_s),Y_s)=(E_1(X_1,Y_1),Z_2,\dots,Z_s)
$$
is $\epsilon'$-close to
$(E_1(X_1,U_{d_1,q}),U_{m_2-d_1,q},\dots,U_{m_s-d_{s-1},q})$.
Then we know that it is also $\epsilon$-close to $(U_{m_1-d_0,q},U_{m_2-d_1,q},\dots,U_{m_s-d_{s-1},q})$ since $E_1$ is a $(k_1,\epsilon_1,q)$ extractor.

Finally, to see that $E$ has degree $t$, note that $E'((X_2,\dots, X_s), Y_s)=(Y_1, Z_2,\dots , Z_s)$ has degree $t'$ in its variables $X_2,\dots, X_s$ and $Y_s$ by the induction hypothesis and hence $Y_1,Z_2,\dots, Z_s$ have degree $t'$ in these variables. Then $Z_1=E_1(X_1,Y_1)$ has degree $t_1\cdot\max\{1, t'\}=t$ in $X_1,\dots, X_s$ and $Y_s$. So $E((X_1, \dots, X_s ), Y_s)=(Z_1, \dots, Z_s )$ has degree $t$ in $X_1,\dots, X_s$ and $Y_s$.
\end{proof}

\againtheorem{thmconverter}{\thmconverter}

We first prove the following technical lemmas.

\begin{lem}[chain rule for min-entropy]\label{lem_chain_rule}
Let $(X,Y)$ be a joint distribution where $X$ is distributed over $\mathbb{F}_q^\ell$ and $Y$ has $q$-ary min-entropy $k$. We have
$$
\Pr_{x\leftarrow X}\left[Y|_{X=x} \text{ has $q$-ary min-entropy } k-\ell-\log_q(1/\epsilon) \right]\geq 1-\epsilon.
$$
\end{lem}

\begin{proof}
We say $x\in\supp(X)$ is {\it good} if $\Pr[X=x]\geq \epsilon q^{-\ell}$ and {\it bad} otherwise. Then $\Pr_{x\leftarrow X}[x \text{ is bad}]\leq\supp(X)\epsilon q^{-\ell}\leq\epsilon$. Consider arbitrary good $x$. For any specific value $y$ for $Y$, we have $\Pr[Y|_{X=x}=y]=\frac{\Pr[(Y=y)\wedge (X=x)]}{\Pr[X=x]}\leq \frac{\Pr[Y=y]}{\epsilon q^{-\ell}}\leq q^{-\lb k-\ell-\log_q(1/\epsilon)\rb}$. By definition, $Y|_{X=x}$ has $q$-ary min-entropy $k-\ell-\log_q(1/\epsilon)$ when $X=x$ is good, which occurs with probability at least $1-\epsilon$.
\end{proof}

\begin{lem}\label{lem_closeness}
Let $P,Q$ be two distributions over set $I$ with $\Delta(P,Q)\leq\epsilon$.
Let $\{X_i: i\in\supp(P)\}$ and $\{Y_i: i\in\supp(Q)\}$ be two collections of distributions over the same set $S$ such that $\Delta(X_i,Y_i)\leq\epsilon'$ for any $i\in\supp(P)\cap\supp(Q)$. Then $X\defeq\sum_{i\in\supp(P)} \Pr[P=i]\cdot X_i$ is $(2\epsilon+\epsilon')$-close to $Y\defeq\sum_{i\in\supp(Q)} \Pr[Q=i]\cdot Y_i$.
\end{lem}
\begin{proof}
Let $T$ be an arbitrary subset of $S$ and we will prove that $|\Pr[X\in T]-\Pr[Y\in T]|\leq 2\epsilon+\epsilon'$.

Note that we can add dummy distributions $X_i$ for $i\in I\setminus\supp(P)$ and $Y_j$ for $j\in I\setminus\supp(Q)$ such that $\Delta(X_i,Y_i)\leq\epsilon'$ for all $i\in I$, and it still holds that $X=\sum_{i\in I}\Pr[P=i]\cdot X_i$ and $Y=\sum_{i\in I}\Pr[Q=i]\cdot Y_i$. Then we have
$$
\begin{aligned}
&|\Pr[X\in T]-\Pr[Y\in T]|\\
&=\left|\sum_{i\in I} \Pr[P=i]\Pr[X_i\in T]-\sum_{i\in I}\Pr[Q=i]\Pr[Y_i\in T]\right|\\
&\leq\sum_{i\in I}|\Pr[P=i]\Pr[X_i\in T]-\Pr[Q=i]\Pr[Y_i\in T]|\\
&\leq\sum_{i\in I}|(\Pr[P=i]-\Pr[Q=i])\Pr[X_i\in T]+\Pr[Q=i](\Pr[X_i\in T]-\Pr[Y_i\in T])|\\
&\leq\lb\sum_{i\in I}|\Pr[P=i]-\Pr[Q=i]|\rb+\epsilon'\lb\sum_{i\in I}\Pr[Q=i]\rb\\
&\leq 2\epsilon+\epsilon'.
\end{aligned}
$$
\end{proof}

\newcommand{\lemalmostblksrc}{
Let $X=(X_1,\dots, X_s)$ be a distribution over $\mathbb{F}_q^{n_1}\times\dots\times\mathbb{F}_q^{n_s}$ such that for any $i\in [s]$ and $(x_1,\dots,x_{i-1})\in\supp(X_1,\dots,X_{i-1})$, the conditional distribution $X_i|_{X_1=x_1,\dots,X_{i-1}=x_{i-1}}$ is $\epsilon$-close to a distribution $\tilde{X}_i(x_1,\dots,x_{i-1})$
with $q$-ary min-entropy $k_i$. Then $X$ is $2s\epsilon$-close to a $(k_1,\dots,k_s)$ $q$-ary block source.
}
\mylemma{lemalmostblksrc}{\lemalmostblksrc}

\begin{proof}
Define $X'=(X'_1,\dots,X'_s)$ as the unique distribution such that for any $i\in [s]$ and any $(x_1,\dots,x_{i-1})\in\supp(X'_1,\dots,X'_{i-1})$, the conditional distribution $X'_i|_{X'_1=x_1,\dots,X'_{i-1}=x_{i-1}}$ equals the distribution $\tilde{X}_i(x_1,\dots,x_{i-1})$ if $(x_1,\dots,x_{i-1})\in\supp(X_1,\dots,X_{i-1})$\footnote{$(x_1,\dots,x_{i-1})\in\supp(X_1,\dots,X_{i-1})$ always holds if $i=1$.} and otherwise equals $U_{n_i,q}$.

For any $i\in [s]$ and $(x_1,\dots,x_{i-1})\in\supp(X'_1,\dots,X'_{i-1})$, we known $X'_i|_{X'_1=x_1,\dots,X'_{i-1}=x_{i-1}}$ is either $\tilde{X}_i(x_1,\dots,x_{i-1})$ or $U_{n_i,q}$. And in either case it has $q$-ary min-entropy $k_i$. So $X'$ is a $(k_1,\dots,k_s)$ $q$-ary block source.

We then prove that for any $i\in [s]$ and $(x_1,\dots,x_{i-1})\in\supp(X_1,\dots,X_{i-1})\cap\supp(X'_1,\dots,X'_{i-1})$, the conditional distribution $X|_{X_1=x_1,\dots,X_{i-1}=x_{i-1}}$ is $2(s-i+1)\epsilon$-close to $X'|_{X'_1=x_1,\dots,X'_{i-1}=x_{i-1}}$. Setting $i=1$ proves the lemma.

Induct on $i$. For $i=s$ the claim holds by the definition of $X'$. For $i<s$, assume the claim holds for $i+1$ and we prove that it holds for $i$ as well. Consider any $(x_1,\dots,x_{i-1})\in\supp(X_1,\dots,X_{i-1})\cap\supp(X'_1,\dots,X'_{i-1})$.
Let $A=X_i|_{X_1=x_1,\dots,X_{i-1}=x_{i-1}}$ and $B=X'_i|_{X'_1=x_1,\dots,X'_{i-1}=x_{i-1}}$.
We have
$$
X|_{X_1=x_1,\dots,X_{i-1}=x_{i-1}}=\sum_{x_i\in\supp(A)}\Pr[A=x_i]\cdot X|_{X_1=x_1,\dots,X_i=x_i}
$$
and
$$
X'|_{X'_1=x_1,\dots,X'_{i-1}=x_{i-1}}=\sum_{x_i\in\supp(B)}\Pr[B=x_i]\cdot X'|_{X'_1=x_1,\dots,X'_i=x_i}.
$$
By the induction hypothesis, we have
$$
\Delta\lb X|_{X_1=x_1,\dots,X_i=x_i}, X'|_{X'_1=x_1,\dots,X'_i=x_i} \rb\leq 2(s-i)\epsilon
$$
for $x_i\in\supp(A)\cap\supp(B)$. Also note that $B$ is identical to $\tilde{X}_i(x_1,\dots,x_{i-1})$ and is $\epsilon$-close to $A$. The claim then follows from \lemref{lem_closeness}.
\end{proof}

Now we are ready to prove \thmref{thmconverter}.

\begin{proof}[Proof of \thmref{thmconverter}]
The degree of $\mathsf{BlkCnvt}_{n,(m_1,\dots,m_s),q}$ is $n$ since $\mathsf{RSCon}_{n,m,q}$ has degree $n$.
Let $X$ be a random source that has $q$-ary min-entropy $k$.
Let $Y_1,\dots,Y_s$ be independent seeds uniformly distributed over $\mathbb{F}_q$.
Let $Z=(Z_1,\dots,Z_s)=\mathsf{BlkCnvt}_{n,(m_1,\dots,m_s),q}(X,(Y_1,\dots,Y_s))$ where each $Z_i=\mathsf{RSCon}_{n,m_i,q}(X,Y_i)$ is distributed over $\mathbb{F}_q^{m_i}$.
Define
$$
B=\left\{(z_1,\dots,z_i):
\begin{aligned}
&i\in [s], ~(z_1,\dots,z_i)\in\supp(Z_1,\dots,Z_i), ~X|_{Z_1=z_1,\dots,Z_i=z_i} \text{ does not }\\
&\text{have $q$-ary min-entropy } k-(m_1+\dots+m_i)-\log_q(1/\epsilon)
\end{aligned}
 \right\}.
$$

Define a new distribution $Z'=(Z'_1,\dots,Z'_s)$ as follows: Sample $z=(z_1,\dots,z_s)\leftarrow Z$ and independently $u=(u_1,\dots,u_s)\leftarrow U_{m_1+\dots,+m_s,q}$. If there exist $i\in [s]$ such that $(z_1,\dots,z_{i-1})\in B$, then pick the smallest such $i$ and let $z'=(z_1,\dots,z_{i-1},u_i,\dots,u_s)$. Otherwise let $z'=z$. Let $Z'$ be the distribution of $z'$.

For any $i\in [s]$ and $(z_1,\dots,z_{i-1})\in\supp(Z'_1,\dots,Z'_i)$, if some prefix of $(z_1,\dots,z_{i-1})$ is in $B$ then $Z'_i|_{Z'_1=z_1,\dots,Z'_{i-1}=z_{i-1}}$ is the uniform distribution $U_{m_i,q}$, otherwise $Z'_i|_{Z'_1=z_1,\dots,Z'_{i-1}=z_{i-1}}=Z_i|_{Z_1=z_1,\dots,Z_{i-1}=z_{i-1}}$. In the second case, $X|_{Z_1=z_1,\dots,Z_{i-1}=z_{i-1}}$ has min-entropy $k-(m_1+\dots+m_{i-1})-\log_q(1/\epsilon)\geq m_i$ since $(z_1,\dots,z_{i-1})\not\in B$. In this case, $Z'_i|_{Z'_1=z_1,\dots,Z'_{i-1}=z_{i-1}}$ is $\epsilon$-close to a distribution of min-entropy $k_i$ by \thmref{thm-condenser} and the fact
$$
Z'_i|_{Z'_1=z_1,\dots,Z'_{i-1}=z_{i-1}}
=Z_i|_{Z_1=z_1,\dots,Z_{i-1}=z_{i-1}}
=\mathsf{RSCon}_{n,m_i,q}(X|_{Z_1=z_1,\dots,Z_{i-1}=z_{i-1}},Y_i).
$$

In either cases $Z'_i|_{Z'_1=z_1,\dots,Z'_{i-1}=z_{i-1}}$ is $\epsilon$-close to a distribution of min-entropy $k_i$. By \lemref{lemalmostblksrc}, $Z'$ is $2s\epsilon$-close to a $(k_1,\dots,k_s)$ $q$-ary block source.

It remains to prove that $Z$ is $s\epsilon$-close to $Z'$, which implies that it is $3s\epsilon$-close to a $(k_1,\dots,k_s)$ $q$-ary block source.
By \lemref{lem_chain_rule}, for any $i\in [s]$, we have
$\Pr[ (Z_1,\dots,Z_{i-1})\in B ]\leq \epsilon$.
So the probability that $(Z_1,\dots,Z_{i-1})\in B$ for some $i\in [s]$ is bounded by $s\epsilon$.
Note that the distribution $Z'$ is obtained from $Z$ by redistributing the weights of $(z_1,\dots,z_s)$ satisfying $(z_1,\dots,z_{i-1})\in B$ for some $i$.
We conclude that $\Delta(Z,Z')\leq s\epsilon$, as desired.
\end{proof}

\againtheorem{thmoutersamp}{\thmoutersamp}

\begin{proof}
We first show that $\mathsf{OuterSamp}_{m,\delta,q}$ is a $(4m,\epsilon,q)$ extractor.
Consider any random source $X$ over $\mathbb{F}_q^n$ with $q$-ary min-entropy $4m$.
Let $s,d_i$ be as in \defnref{defn_outer_samp}. Let $k_i=4\cdot 0.99\cdot d_i$ for $i\in [s]$.
Let $\epsilon_0=\frac{\epsilon}{4s}$.

We have $\lb\sum_{i=1}^s 4d_i\rb+\log_q(1/\epsilon_0)
\leq 4m$ for sufficiently large $q\geq(n/\epsilon)^{O(1)}$.
So by \thmref{thmconverter}, $\mathsf{BlkCnvt}_{n,(4d_1,\dots,4d_s),q}$ is a $(4m,(k_1,\dots,k_s),3s\epsilon_0,q)$ block source converter.
Therefore the distribution $\mathsf{BlkCnvt}_{n,(4d_1,\dots,4d_s),q}(X,U_{s,q})$ is $3s\epsilon_0$-close to a $(k_1,\dots,k_s)$ $q$-ary block source $X'$. Then $\mathsf{OuterSamp}_{m,\delta,q}(X,U_{d,q})$ is $3s\epsilon_0$-close to
$\mathsf{BlkExt}(\mathsf{Line}_{2,q^{d_1}},\dots,\mathsf{Line}_{2,q^{d_s}})(X',U_{1,q})$.

By \lemref{lem_line_extractor}, $\mathsf{Line}_{2,q^{d_i}}$ is a $\lb k_i/d_i,\epsilon_0,q^{d_i}\rb$ extractor for $i\in [s]$ since $3+3\log_{q^{d_i}}(1/\epsilon_0)\leq 4\cdot 0.99=k_i/d_i$. Equivalently it is a $(k_i,\epsilon_0,q)$ extractor.
By \lemref{lemblockextraction}, $\mathsf{BlkExt}(\mathsf{Line}_{2,q^{d_1}},\dots,\mathsf{Line}_{2,q^{d_s}})$
is a $((k_1,\dots,k_s),s\epsilon_0,q)$ block source extractor.
Therefore  $\mathsf{BlkExt}(\mathsf{Line}_{2,q^{d_1}},\dots,\mathsf{Line}_{2,q^{d_s}})(X',U_{1,q})$ is $s\epsilon_0$-close to $U_{m,q}$, which by the previous paragraph, implies that $\mathsf{OuterSamp}_{m,\delta,q}(X,U_{d,q})$ is $4s\epsilon_0$-close to $U_{m,q}$.
By definition, $\mathsf{OuterSamp}_{m,\delta,q}$ is a $(4m,\epsilon,q)$ extractor.
By \thmref{equiv}, it is also an $(\epsilon,\delta)$ sampler.

We have $d=s+1=O(\log m)$ and $n=O\lb m+\log_q(1/\delta)\rb$.
By \lemref{lemdegree}, each $\mathsf{Line}_{2,q^{d_i}}$ has degree $2$ as a manifold over $\mathbb{F}_q$.
Therefore by \lemref{lemblockextraction}, $\mathsf{BlkExt}(\mathsf{Line}_{2,q^{d_1}},\dots,\mathsf{Line}_{2,q^{d_s}})$ has degree $2^s$. By \thmref{thmconverter}, $\mathsf{BlkCnvt}_{n,(4d_1,\dots,4d_s),q}$ has degree $n$.  Therefore $\mathsf{OuterSamp}_{m,\delta,q}$ has degree $n2^s=O\lb m^2+m\log_q(1/\delta)\rb$.
\end{proof}

\againlemma{lemresampling}{\lemresampling}

\begin{proof}
Consider an arbitrary subset $A\subseteq\mathbb{F}_q^{d_0}$.
Define $B(x)=\left\{z\in\mathbb{F}_q^{d_2}: S_1(x,z)\in A\right\}$ for each $x\in\mathbb{F}_q^{n_1}$.
Pick $x_1\leftarrow U_{n_1,q}$ and $x_2\leftarrow U_{n_2,q}$. If $\left|\mu_{S_1\circ S_2((x_1,x_2))}(A)-\mu(A)\right|>\epsilon_1+\epsilon_2$ occurs, then either $|\mu_{S_1(x_1)}(A)-\mu(A)|>\epsilon_1$, or
$|\mu_{S_1\circ S_2((x_1,x_2))}(A)-\mu_{S_1(x_1)}(A)|>\epsilon_2$ occurs. Call the two events $E_1$ and $E_2$ respectively.

Note that $E_1$ occurs with probability at most $\delta_1$ by the sampling property of $S_1$.
Also note that
$$
\mu_{S_1\circ S_2((x_1,x_2))}(A)=\Pr_y[S_1(x_1,S_2(x_2,y))\in A]=\Pr_y[S_2(x_2,y)\in B(x_1)]=\mu_{S_2(x_2)}(B(x_1))
$$
whereas
$$
\mu_{S_1(x_1)}(A)=\Pr_y[S_1(x_1,y)\in A]=\Pr_y[y\in B(x_1)]=\mu(B(x_1)).
$$
So the probability that $E_2$ occurs is
$
\Pr_{x_1,x_2}[|\mu_{S_2(x_2)}(B(x_1))-\mu(B(x_1))|>\epsilon_2]
$ which is bounded by $\delta_2$ by the sampling property of $S_2$.
By the union bound, the event
$$
\left|\mu_{S_1\circ S_2((x_1,x_2))}(A)-\mu(A)\right|>\epsilon_1+\epsilon_2
$$
occurs with probability at most $\delta_1+\delta_2$, as desired.

Finally, we have $S_1\circ S_2((X_1,X_2),Y))=S_1(X_1,S_2(X_2,Y))$ which has degree $t_1 t_2$ in its variables $X_1,X_2,Y$ since $S_1$ and $S_2$ have degree $t_1$ and $t_2$ respectively.
\end{proof}

%\section{Figures}
%
%In this section we use figures to illustrate the data flow in the outer sampler and inner sampler constructed in this paper. For a sampler $S:\mathcal{N}\times\mathcal{D}\to\mathcal{M}$, the input from the top in the corresponding figure denotes the randomness $x\in\mathcal{N}$ which specifies a sample $\{S(x,y): y\in \mathcal{D}\}$. The input from the left denotes an index $y\in \mathcal{D}$. And the output $Z$ denotes the element $S(x,y)$ in the sample.
%

\end{document}